\documentclass[12pt]{iopart}
\usepackage{amsfonts,amsmath,amssymb,amsthm}
\usepackage{bm,epsfig,revsymb,upgreek,url}
\usepackage[colorlinks=true,dvipdfm]{hyperref}

\newcommand{\Cbb}{\mathbb{C}}
\newcommand{\Ccal}{\mathcal{C}}
\newcommand{\Dcal}{\mathcal{D}}
\newcommand{\Ecal}{\mathcal{E}}
\newcommand{\Fcal}{\mathcal{F}}
\newcommand{\grad}{\mathrm{grad} \,}

\newcommand{\Hc}{\mathcal{H}_{\rmc}}
\newcommand{\Hcal}{\mathcal{H}}
\newcommand{\He}{\mathcal{H}_{\rme}}
\newcommand{\Hm}{\mathrm{H}}
\newcommand{\Hs}{\mathcal{H}_{\rms}}
\newcommand{\Ic}{\openone_{\rmc}}
\newcommand{\Ie}{\openone_{\rme}}
\newcommand{\Imag}{\mathrm{Im}}
\newcommand{\Is}{\openone_{\rms}}
\newcommand{\J}{\mathcal{J}_{\tf}^{[V]}}
\newcommand{\Jsys}{\mathcal{J}_{\tf}^{[\Vsys]}}
\newcommand{\Kbb}{\mathbb{K}}
\newcommand{\Kcal}{\mathcal{K}}
\newcommand{\Lcal}{\mathcal{L}}
\newcommand{\Mcal}{\mathcal{M}}
\newcommand{\Ncal}{\mathcal{N}}
\newcommand{\Qcal}{\mathcal{Q}}
\newcommand{\Rbb}{\mathbb{R}}
\newcommand{\Real}{\mathrm{Re}}
\newcommand{\rmc}{\mathrm{c}}
\newcommand{\rHS}{\mathrm{HS}}
\newcommand{\rmint}{\mathrm{int}}
\newcommand{\rmk}{\mathrm{k}}
\newcommand{\rml}{\mathrm{l}}
\newcommand{\rmm}{\mathrm{m}}
\newcommand{\rms}{\mathrm{s}}
\newcommand{\rmT}{\mathrm{T}}
\newcommand{\rmu}{\mathrm{u}}
\newcommand{\Scal}{\mathcal{S}}
\newcommand{\tf}{t_{\mathrm{f}}}
\newcommand{\Tre}{\mathrm{Tr_{\rme}}}
\newcommand{\Trs}{\mathrm{Tr_{\rms}}}
\newcommand{\Uc}{\mathrm{U}(\Hc)}
\newcommand{\Ue}{\mathrm{U}(\He)}
\newcommand{\Uenv}{U_{\rme}}
\newcommand{\Us}{\mathrm{U}(\Hs)}
\newcommand{\Usys}{U_{\rms}}
\newcommand{\Venv}{V_{\rme}}
\newcommand{\Vsys}{V_{\rms}}

\newtheorem{definition}{Definition}
\newtheorem{lemma}{Lemma}
\newtheorem{theorem}{Theorem}

\begin{document}

\title[Environment-invariant measure of distance between
 evolutions]{Environment-invariant measure of distance between
 evolutions of an open quantum system}

\author{Matthew~D.~Grace$^{1}$, Jason~Dominy$^{2}$,
 Robert~L.~Kosut$^{3}$, Constantin~Brif\,$^{4}$ and
 Herschel~Rabitz$^{4}$}

\address{$^{1}$Department of Scalable Computing Research \& Development,
 Sandia National Laboratories, Livermore, CA 94550}

\address{$^{2}$Program in Applied \& Computational Mathematics,
 Princeton University, Princeton, New Jersey 08544}

\address{$^{3}$SC Solutions, Inc., Sunnyvale, CA 94085}

\address{$^{4}$Department of Chemistry, Princeton University, Princeton,
 New Jersey 08544}

\eads{\mailto{mgrace@sandia.gov}, \mailto{jdominy@math.princeton.edu},
 \mailto{kosut@scsolutions.com}, \mailto{cbrif@princeton.edu}, and
 \mailto{hrabitz@princeton.edu}}

\date{\today}

\begin{abstract}
The problem of quantifying the difference between evolutions of an open
quantum system (in particular, between the actual evolution of an open
system and the ideal target operation on the corresponding closed
system) is important in quantum control, especially in control of
quantum information processing. Motivated by this problem, we develop
a measure for evaluating the distance between unitary evolution
operators of a composite quantum system that consists of a sub-system of
interest (e.g., a quantum information processor) and environment. The
main characteristic of this measure is the invariance with respect to
the effect of the evolution operator on the environment, which follows
from an equivalence relation that exists between unitary operators
acting on the composite system, when the effect on only the sub-system
of interest is considered. The invariance to the environment's
transformation makes it possible to quantitatively compare the evolution
of an open quantum system and its closed counterpart. The distance
measure also determines the fidelity bounds of a general quantum channel
(a completely positive and trace-preserving map acting on the sub-system
of interest) with respect to a unitary target transformation. This
measure is also independent of the initial state of the system and
straightforward to numerically calculate. As an example, the measure is
used in numerical simulations to evaluate fidelities of optimally
controlled quantum gate operations (for one- and two-qubit systems), in
the presence of a decohering environment. This example illustrates the
utility of this measure for optimal control of quantum operations in the
realistic case of open-system dynamics.
\end{abstract}


\submitto{\NJP}

\maketitle

\section{Introduction}
\label{sec:intro}

Control and optimization of open quantum systems is a rapidly growing
area of theoretical \cite{Lloyd00a, Lloyd01a, Brif01a, Zhu03a,
Altafini03a, Altafini04a, Sklarz04a, Khodjasteh05a, Romano06a,
Jirari06a, Wenin06a, Li06a, Wu07a, Mohseni09a, Rabitz09a} and
experimental \cite{Brixner01a, Herek02a, Felinto05a, Morton06a,
Brander08a, Walle09a, Biercuk09a, Roth09a} research. This research is
connected, to a considerable degree, to a large body of work that has
been recently devoted to integrating concepts and methods of optimal
quantum control \cite{Rabitz00a, Rabitz00b, Walmsley03a, Balint08a} into
the fields of quantum information and quantum computation (QC)
\cite{Nielsen00a}. In particular, optimal control theory was used to
design external control fields for generating quantum-gate operations in
the presence of an environment \cite{Grig05a, Schulte06a, Hohen06a,
Grace07a, Grace07b, Wenin08a, Wenin08b, Reben09a}. In this context, an
important problem is to evaluate the distance between different quantum
operations, especially between actual and ideal processes
\cite{Nielsen00a, Uhlmann76a, Jozsa94a, Schumach96a, Fuchs99a,
Kretsch04a, Gilchrist05a, Lidar08a}. One obstacle in the comparison of
such processes is the reliance of many existing distance and fidelity
measures on the initial state of the system or, to a lesser extent, on
an integration or optimization over all possible states. Another
difficulty arises when we consider the evolution of the quantum system
of interest, in this case, a quantum information processor (QIP), that
is coupled to a (finite-dimensional) environment. The main obstacle is
that the unitary time-evolution operator for the composite system
(consisting of the QIP and environment) and the target unitary operator
for the QIP have different dimensions.

In this work, we develop a measure (originally proposed in
\cite{Kosut06a}) that quantifies the distance between evolutions of a
composite (bipartite) quantum system in such a way that only the
effect of the evolution operator on one component (the sub-system of
interest) is relevant in the measure, while the effect on the other
component (the environment) can be arbitrary. Due to the invariance to
the transformation on the environment, it is possible to evaluate the
difference between unitary evolutions of the composite system (including
the environment) and the sub-system of interest alone (excluding the
environment), i.e., between unitary quantum operations of different
dimensions. Therefore, the proposed measure enables us to compare
unitary quantum operations (or, more generally, quantum channels
\cite{Nielsen00a}) corresponding to dynamics of ideal (i.e., closed) and
real (i.e., open) systems. When the sub-system of interest is not
coupled to an environment, the proposed distance measure reduces to a
form that was previously used to evaluate the fidelity of closed-system
unitary operations (see, e.g., \cite{Palao02a, Palao03a, Sporl07a}).
Another important advantage of this distance measure is that it is
independent of the initial state of the system. This property is
important, since the fundamental objective of QC is to generate a
specified target transformation for \emph{any} initial state.

The development of an appropriate objective functional is a crucial
issue for control and optimization (for both classical and quantum
systems) \cite{Stengel94a}. Since maxima or minima of the objective
functional determine control optimality, the selection of the objective
is just as important as that of the physical model and optimization
algorithm. Therefore, constructing an objective functional that
represents control goals in the best possible way is one of the first
tasks in any application of optimal control. In addition, the emerging
studies of quantum control landscapes \cite{Chakr07a}, for
state-preparation control \cite{Rabitz04a, Rabitz06a}, observable
control \cite{Rabitz06b, Wu08a, Hsieh09a}, and unitary-transformation
control \cite{Rabitz05a, Hsieh08a, Ho09a}, emphasize the profound
effects that the choice of the objective functional has on the
optimization [the control landscape is defined by the objective as a
function of control variables, $\mathcal{J} = \mathcal{J}(C)$]. In
particular, the relationship between the control landscape structure and
optimization search efforts was recently analyzed \cite{Moore08a,
Oza09a}. The measure we present here is well-suited for directly
evaluating distances between general operations resulting from control
of open quantum systems. The utility of this distance measure for
optimal control of quantum gates in the presence of an environment is
illustrated with a detailed numerical example.

This article is organized as follows. Section~\ref{sec:equiv} presents
an equivalence relation between unitary evolution operators acting on a
composite (bipartite) quantum system. This relation is based on the
effect of these operators on the sub-system of interest, while the
effect on the other sub-system (representing the environment) is not
taken into account and thus can be arbitrary. The equivalence relation
produces a quotient space that is a homogeneous manifold of the original
space of unitary operators. In section~\ref{sec:quot-metric}, a quotient
metric is defined on this space, resulting in a generalized expression
of the distance measure. This measure can be used to quantitatively
compare unitary evolutions of the entire composite system and one of its
components (the sub-system of interest), e.g., unitary operations of
different dimensions. Particular forms of the general distance measure
can be obtained using different matrix norms. In section~\ref{sec:HS},
the distance measure is computed explicitly using the Hilbert-Schmidt
norm, producing an analytical expression that is calculated directly
from the unitary operator of the composite system. In addition, several
properties of this specific measure are discussed.
Section~\ref{sec:singular} develops the distance measure using the
induced two-norm (i.e., the maximum singular value). In
section~\ref{sec:fidelity}, we relate the distance measure developed in
section~\ref{sec:HS} to a few typical fidelity measures defined for a
general quantum channel (acting on the sub-system of interest) with
respect to a target unitary transformation. It is shown that the
distance measure and an average channel fidelity both depend on the
trace norm of the same matrix (as the matrix norm increases, distance
decreases and fidelity increases), and that the distance measure sets
the lower and upper bounds of this fidelity. Section~\ref{sec:control}
illustrates, using a numerical example, how the distance measure
developed here can be directly incorporated into the optimal control of
quantum operations in the presence of a decohering environment. In this
application, the actual unitary evolution operator for the composite
system (consisting of a QIP and its environment) is compared to the
target unitary operator for the QIP alone, which is of a smaller
dimension. This example demonstrates the utility of the distance measure
in optimal control of unitary quantum gates for QC. Lastly,
section~\ref{sec:conclusion} concludes the paper with a summary of the
results and a brief discussion of future directions. For the continuity
of the presentation, some mathematical details and most proofs are
relegated to appendices.

\section{Equivalence classes of unitary evolution operators}
\label{sec:equiv}

Consider a finite-dimensional, closed, bipartite quantum system with a
Hamiltonian (possibly time-dependent) of the form
\begin{equation}
\label{eq:Hamiltonian1}
H := H_{\rms} + H_{\rme} + H_{\rmint},
\end{equation}
where $H_{\rms}$ and $H_{\rme}$ represent the Hamiltonians for the
system $\Scal$ (e.g., a QIP) and environment $\Ecal$, respectively, and
$H_{\rmint}$ represents the system-environment interaction. Let $\Hs$
and $\He$ denote the respective Hilbert spaces of the system and
environment, where $n_{\rms} := \dim\{ \Hs \}$ and $n_{\rme} := \dim\{
\He \}$. Thus, the Hilbert space of the composite system $\Ccal$ is $\Hc
:= \Hs \otimes \He$ and $n := \dim \{ \Hc \} = n_{\rms} n_{\rme}$. Let
$\{ |i\rangle \}$, $\{ |\nu\rangle \}$ and $\{ |i\rangle \otimes
|\nu\rangle \}$ be orthonormal bases that span the Hilbert spaces $\Hs$,
$\He$, and $\Hc$, respectively.

For a given Hilbert space $\Hcal$, the set of admissible states,
represented as density matrices on $\Hcal$, is denoted by $\Dcal
(\Hcal)$. As such, a density matrix $\rho \in \Dcal (\Hcal)$ is a
positive operator of trace one, i.e., $\rho \geq 0$ and $\Tr(\rho) =
1$. Elements of $\Dcal (\Hcal)$ correspond to either pure $\left[ \Tr
\left( \rho^{2} \right) = 1 \right]$ or mixed $\left[ \Tr \left(
\rho^{2} \right) < 1 \right]$ states. Density matrices for pure states
can also be expressed as $|\psi\rangle \langle\psi|$, where
$|\psi\rangle$ represents a normalized vector in $\Hcal$ (i.e.,
$\langle\psi|\psi\rangle$ = 1). The group of all unitary operators on
$\Hcal$ is denoted by $\mathrm{U}(\Hcal)$.

Let $U(t) \in \Uc$ denote the unitary time-evolution operator of the
composite system $\Ccal$, whose evolution is governed by the
Schr\"{o}dinger equation ($\hbar = 1$):
\begin{equation} 
\label{eq:Schrodinger} 
\frac{\rmd U(t)}{\rmd t} = - \rmi H(t) U(t),
\end{equation} 
with the initial condition $U(t = 0) = \Ic$, the identity operator on
$\Hc$. Given an initial density matrix $\rho(t = 0) \in
\Dcal(\Hc)$, this state evolves in time as
\begin{equation}
\label{eq:density-evolution}
\rho(t) = U(t)\rho(0) U^{\dag}(t).
\end{equation}
For any state $\rho \in \Dcal(\Hc)$, the state of the system $\Scal$,
with which system observables are calculated, is described by the
reduced density matrix $\rho_{\rms} := \Tre(\rho)$, where $\Tre$ denotes
the partial trace over $\He$. Thus, from (\ref{eq:density-evolution}),
$\rho_{\rms}$ evolves as
\begin{equation}
\rho_{\rms}(t) = \Tre \big[ U(t) \rho(0) U^{\dag}(t) \big].
\end{equation}
If two composite unitary operators $U, V \in \Uc$ are such that 
\begin{equation}
\label{eq:equiv}
\Tre(U \rho U^{\dag}) = \Tre( V \rho V^{\dag}), \ \ \forall \rho \in
\Dcal(\Hc),
\end{equation}
then they will produce identical evolutions of the reduced density
matrix $\rho_{\rms}$. In this situation, $U$ and $V$ are said to be
\emph{system}- or $\Scal$-\emph{equivalent}, a relationship expressed as
$U \sim V$. The condition described by (\ref{eq:equiv}) defines an
equivalence relation \cite{Naylor82a} (in addition to being symmetric
and reflexive, $\Scal$-equivalence is also transitive, i.e., if $U_{1}
\sim U_{2}$ and $U_{2} \sim U_{3}$, then $U_{1} \sim U_{3}$). The
necessary and sufficient condition for $\Scal$-equivalence is stated in
the following theorem.

\begin{theorem}
\label{thm:equiv} 
Let $U, V \in \Uc$. Then $\Tre \left( U \rho U^{\dag}
\right) = \Tre \left( V \rho V^{\dag} \right)$ for all density matrices
$\rho \in \Dcal(\Hc)$ if and only if $U = (\Is \otimes \Phi)V$ for
some $\Phi \in \Ue$.
\end{theorem}

\begin{proof}
See \ref{app:equiv}.
\end{proof}

This theorem is related to a unitary invariance of Kraus maps
\cite{Nielsen00a, Kraus83a}, which is described further in
section~\ref{sec:fidelity}.

For any $U \in \Uc$, let $[U]$ denote the equivalence class of $U$,
i.e., $[U] = \{ (\Is \otimes \Phi) U : \Phi \in \Ue \}$. As a
consequence of Theorem~\ref{thm:equiv}, every equivalence class is
isomorphic to the closed subgroup $\Is \otimes \Ue = \{ \Is \otimes
\Phi : \Phi \in \Ue \}$, i.e., $[U]$ is just a right-translation of
$\Is \otimes \Ue$ by $U$. Therefore, if only the evolution of $\Scal$ is
relevant, then just the space of $\Scal$-equivalence classes, i.e., the
\emph{quotient space} $\Qcal := \Uc / \Is \otimes \Ue$ needs to be
considered, rather than the space of all unitary operators $\Uc$. While
$\Qcal$ is not a Lie group like $\Uc$, it is a smooth (right)
homogeneous manifold of dimension $(n_{\rms}^{2}-1) n_{\rme}^{2}$
\cite{Warner83a}. Although not explored here, the fact that $\Qcal$ has
the additional structure of a homogeneous manifold may be useful for
further analysis (e.g., explorations of the landscape topology
\cite{Ho09a}) of the distance measure.

\section{Properties of metrics on the quotient space of
equivalent operators}
\label{sec:quot-metric}

A \emph{metric space} consists of a pair of objects: a set $\Omega$ and
a real-valued \emph{metric} or \emph{distance measure} $\Delta$. The
ordered pair $(\Omega, \Delta)$ is a metric space if the following
conditions are satisfied for any $x, y, z \in \Omega$ \cite{Naylor82a}:
\begin{enumerate}
\item{Non-negativity: $\Delta(x, y) \geq 0$}
\item{Strict positivity: $\Delta(x, y) = 0$ if and only if $x = y$}
\item{Symmetry: $\Delta(x, y) = \Delta(y, x)$}
\item{Triangle inequality: $\Delta(x, z) \leq \Delta(x, y) + \Delta(y,
 z)$}
\end{enumerate}

In this work, $\Omega$ is the quotient space $\Qcal$ developed in
section~\ref{sec:equiv}; elements of this space are the equivalence
classes $[U]$. Let $\Delta$ be any metric on $\Uc$, then the
corresponding \emph{quotient metric} $\tilde{\Delta}$ on $\Qcal$ is
\begin{equation}
\label{eq:quot-metric}
\tilde{\Delta} ([U], [V]) := \min_{\Phi_{1}, \Phi_{2}} \Delta
\big( (\Is \otimes \Phi_{1}) U, (\Is \otimes \Phi_{2}) V \big).
\end{equation}
This expression can be used to obtain the quotient metric for various
choices of $\Delta$, for example, when $\Delta$ is the Hilbert-Schmidt
norm distance, the induced two-norm distance, the geodesic distance with
respect to some Riemannian metric, etc. \cite{Reed80a, Horn90a,
doCarmo92a}.

The quotient metric (\ref{eq:quot-metric}) does not involve the partial
trace over the environment, which is implicit in other distance (and
fidelity) measures that use the operator-sum representation
\cite{Nielsen00a, Jozsa94a, Schumach96a, Fuchs99a, Kretsch04a,
Gilchrist05a, Lidar08a}. Moreover, the quotient metric is independent of
any particular initial state, which is a very useful practical feature
in applications to QC. In contrast to other measures of quantum gate
fidelity, e.g., some of those in \cite{Nielsen00a, Jozsa94a,
Schumach96a, Fuchs99a, Kretsch04a, Gilchrist05a, Lidar08a},
$\tilde{\Delta} ([U], [V])$ is directly evaluated by propagating the
evolution operator $U(t)$ in (\ref{eq:Schrodinger}), without making any
assumptions about the initial state of the composite system
$\Ccal$. This property allows one to consistently quantify the distance
between unitary quantum operations (e.g., the actual and ideal
operation) for \emph{any} initial state.

If $\Delta$ is a left-invariant metric on $\Uc$ [i.e., $\Delta(WU, WV)
= \Delta(U, V)$ for any $U, V, W \in \Uc$], then
(\ref{eq:quot-metric}) becomes 
\begin{equation}
\label{eq:quot-metric-left}
\tilde{\Delta} ([U], [V]) = \min_{\Phi} \Delta \big( U,(\Is
\otimes \Phi) V \big).
\end{equation}
In the context of a norm distance, (\ref{eq:quot-metric-left}) can also
be expressed as
\begin{equation}
\label{eq:norm-dist}
\tilde{\Delta} ([U], [V]) = \lambda \min_{\Phi} \| U - (\Is \otimes
\Phi) V \|,
\end{equation}
where $\lambda$ is a specified normalization factor and $\| \cdot \|$
is any left-invariant matrix norm on the space of $n \times n$ complex
matrices, denoted as $\Mcal_{n} \left( \Cbb \right)$.

In the following sections, we introduce some useful properties of
unitarily invariant quotient metrics on the homogeneous space $\Qcal$.

\subsection{Closed-system target transformations}
\label{ssec:closedsystem-target}

A very typical situation in QC occurs when the target unitary
transformation is specified for an ideal closed system (i.e., for
$\Scal$ alone, excluding coupling to the environment). Let $U$ denote
the actual evolution operator for the composite system $\Ccal$ and $V$
denote the target quantum operation. If the target is defined for the
closed system, then $V$ is decoupled: $V = \Vsys \otimes \Venv$, where
$\Vsys \in \Us$ is the specified target unitary transformation for
$\Scal$ and $\Venv \in \Ue$ is arbitrary. Because $\Venv$ is an
arbitrary operation on $\Ecal$, it can be incorporated into $\Phi_{2}$
in (\ref{eq:quot-metric}):
\begin{displaymath}
\min_{\Phi_{1}, \Phi_{2}} \Delta \big( (\Is \otimes \Phi_{1}) U, (\Is
\otimes \Phi_{2}) (\Vsys \otimes \Venv) \big)
= \min_{\Phi_{1}, \Phi_{2}} \Delta \big( (\Is \otimes \Phi_{1}) U, (\Is
\otimes \Phi_{2}) (\Vsys \otimes \Ie) \big).
\end{displaymath}
Similarly, $\Venv$ can be incorporated into $\Phi$ in
(\ref{eq:quot-metric-left}) and (\ref{eq:norm-dist}). Therefore, for
the closed-system target $\Vsys$, expressions (\ref{eq:quot-metric}),
(\ref{eq:quot-metric-left}), and (\ref{eq:norm-dist}) of the quotient
metric are reformulated as
\begin{subequations}
\begin{align}
& \tilde{\Delta} ([U], [\Vsys \otimes \Ie]) = 
\min_{\Phi_{1}, \Phi_{2}} \Delta \big( (\Is \otimes \Phi_{1}) U, 
(\Vsys \otimes \Phi_{2}) \big), \\
& \tilde{\Delta} ([U], [\Vsys \otimes \Ie]) = 
\min_{\Phi} \Delta \big( U,(\Vsys \otimes \Phi) \big), \\
\intertext{and}
& \tilde{\Delta} ([U], [\Vsys \otimes \Ie]) = 
\lambda \min_{\Phi} \| U - (\Vsys \otimes \Phi) \|,
\end{align}
\end{subequations}
respectively. From a practical point of view, these results mean
that we can use the distance measure $\tilde{\Delta} ([U], [\Vsys
\otimes \Ie])$ to quantitatively compare the actual evolution
operator $U \in \Uc$ of the composite system $\Ccal$ (of dimension $n$)
and the target transformation $\Vsys \in \Us$ specified for the
sub-system $\Scal$ (of a smaller dimension $n_{\rms}$).

\subsection{Chaining of unitary operations}
\label{ssec:chaining}

Consider a physical process composed of several components, such as
the application of a sequence of $n$ unitary operations, $U_{n} \ldots
U_{2} U_{1}$, acting on the composite system, compared to the
corresponding sequence of target unitary operations, $V_{n} \ldots
V_{2} V_{1}$, where $U_{i}, V_{i} \in \Uc$. In this context, a metric
or distance measure is said to possess the \emph{chaining} criterion
\cite{Gilchrist05a} if
\begin{equation}
\label{eq:chain}
\Delta \left( U_{m} \ldots U_{2} U_{1}, V_{m} \ldots V_{2} V_{1} \right)
\leq \Delta \left( U_{1}, V_{1} \right) + \Delta \left( U_{2}, V_{2}
\right) + \ldots + \Delta \left( U_{m}, V_{m} \right).
\end{equation}
Thus, the total error of the process is bounded from above by the sum of
the individual errors $\Delta(U_{i}, V_{i})$. This property is useful
for analyzing the overall fidelity of a quantum algorithm, consisting
of, e.g., a sequence of one- and two-qubit operations \cite{Nielsen00a,
 Gilchrist05a}.

\begin{theorem}
If $\Delta$ is a bi-invariant metric on $\Uc$ [i.e., $\Delta(W_{1} U
W_{2}, W_{1} V W_{2}) = \Delta(U, V)$ for any $U, V, W_{1}, W_{2} \in
\Uc$], then the chaining criterion (\ref{eq:chain}) holds for any
$U_{1}, U_{2}, \ldots, U_{m} \in \Uc$ and $V_{1}, V_{2}, \ldots, V_{m}
\in \Us \otimes \Ue$.
\end{theorem}

\begin{proof}
See \ref{app:chain}.
\end{proof}

From this result, the chaining property holds for $V_{i} \in \Us \otimes
\Ue$, i.e., the set of all decoupled or factorizable unitary operators
in $\Uc$. This condition is satisfied by any $V_{i}$ that preserves the
isolation of the system $\Scal$ from the environment $\Ecal$, a
ubiquitous objective in QC.

\subsection{Convexity}
\label{ssec:convexity}

Performing the minimization in (\ref{eq:quot-metric}) over $\Phi_{1},
\Phi_{2} \in \Ue$ to compute $\tilde{\Delta} ([U], [V])$ may be
difficult and expensive in general. However, when the distance measure
is determined by a left-invariant norm (\ref{eq:quot-metric-left}), the
minimization can be recast as a convex optimization problem
\cite{Boyd04a}. Relaxing the condition on $\Phi$ from $\Phi^{\dag} \Phi
= \Ie$ to $\Phi^{\dag} \Phi \leq \Ie$ produces such a problem since (a)
any norm is a convex function, (b) the argument in the norm is affine in
the optimization variable $\Phi$, and (c) the set of matrices that
satisfies $\Phi^{\dag} \Phi \leq \Ie$ is a convex set in $\Phi$. To show
this, note that the set of matrices that satisfies $\Phi^{\dag} \Phi
\leq \Ie$ is equal to the set that satisfies $\| \Phi \|_{2} \leq 1$,
where $\| \cdot \|_{2}$ is the induced two-norm described in
section~\ref{sec:singular}. Thus, $\| \theta \Phi_{1} + (1 - \theta)
\Phi_{2} \|_{2} \leq \theta \| \Phi_{1} \|_{2} + (1 - \theta) \|
\Phi_{2} \|_{2} \leq 1$ for all $0 \leq \theta \leq 1$, which follows
from the triangle inequality. Expressing (\ref{eq:quot-metric-left})
[or (\ref{eq:norm-dist})] as a convex optimization problem is desirable,
provided the solution is on the boundary of $\Phi^{\dag} \Phi \leq \Ie$
(i.e., $\Phi \in \Ue$), because if a local minimum exists in a convex
set, it is also a global minimum \cite{Boyd04a}.

\subsection{Stability}
\label{ssec:stability}

Another important metric property is \emph{stability}
\cite{Gilchrist05a}. A metric is said to be \emph{stable} if
\begin{equation}
\label{eq:stable}
\Delta([U] \otimes \Xi, [V] \otimes \Xi) = \Delta([U], [V]),
\end{equation}
for any unitary operator $\Xi$ of any given dimension. Stability means
that operations on an ancillary system that is not coupled to the
original composite system do not affect the value of $\Delta$, and
hence $\tilde{\Delta}$ \cite{Gilchrist05a}. We do not prove stability
for the general form of the distance measure in (\ref{eq:quot-metric})
or (\ref{eq:norm-dist}), but demonstrate it for the Hilbert-Schmidt
norm distance in the following section.

\section{Computing the distance measure with the
Hilbert-Schmidt norm}
\label{sec:HS}

In this section, the distance measure presented in
section~\ref{sec:quot-metric} is evaluated with the Hilbert-Schmidt
norm (also referred to as the Frobenius norm), which is defined as
\cite{Horn90a}
\begin{equation}
\label{eq:HS-norm}
\| M \|_{\rHS} := \sqrt{\Tr \left( M^{\dag} M \right)} = \left[ \sum_{i
 = 1}^{n} \sigma_{i}^{2}(M) \right]^{1/2}, \ \ \forall M \in \Mcal_{n}
\left( \Cbb \right),
\end{equation}
where $\sigma_{i}(M)$ is the $i$th singular value of $M$. The $i$th
singular value, $\sigma_{i}(M)$, is the square root of the $i$th
eigenvalue of the positive square matrix $M^{\dag} M$, where the
singular values appear in descending order: $\sigma_{i}(M) \geq
\sigma_{i+1}(M)$ for all $i$.

The general form of the Hilbert-Schmidt norm distance is, according to
(\ref{eq:norm-dist}), 
\begin{equation}
\label{eq:HS-dist-a}
\tilde{\Delta}_{\rHS} ([U], [V]) := \lambda_{n} \min_{\Phi}
\Big\{ \| U - (\Is \otimes \Phi) V \|_{\rHS} \ : \ \Phi \in
\Ue \Big\}.
\end{equation}
Using (\ref{eq:HS-norm}) in (\ref{eq:HS-dist-a}) and $\lambda_n =
(2n)^{-1/2}$, the minimization over $\Phi$ can be carried out
analytically and, as shown below, we obtain:
\begin{equation}
\label{eq:HS-dist-b} 
\tilde{\Delta}_{\rHS} ([U], [V]) = \left( 1 - \frac{1}{n} 
\| \Gamma \|_{\Tr} \right)^{1/2},
\end{equation}
where $\Gamma \in \Mcal_{n_{\rme}}(\Cbb)$ is
\begin{equation}
\label{eq:gamma} 
\Gamma := \Trs \left( U V^{\dag} \right),
\end{equation}
and $\| \cdot \|_{\Tr}$ in (\ref{eq:HS-dist-b}) is the trace norm
\cite{Horn90a}:
\begin{equation}
\label{eq:trace-norm}
\| M \|_{\Tr} := \Tr \left( \sqrt{M^{\dag} M} \right) =
\sum_{i = 1}^{n} \sigma_{i}(M), \ \ \forall M \in \Mcal_{n} (\Cbb).
\end{equation}
Note that in (\ref{eq:gamma}), $\Trs( \cdot )$ denotes the partial trace
over the system space $\Hs$.

\begin{proof}
The derivation of (\ref{eq:HS-dist-b}) begins with some algebraic
rearrangement of (\ref{eq:HS-dist-a}), initially yielding
\begin{subequations}
\begin{align}
\tilde{\Delta}_{\rHS} ([U], [V]) & = \min_{\Phi} \sqrt{1 -
 \frac{1}{n} \Real \bigg\{ \Tr \Big[ U V^{\dag} \left(
 \Is \otimes \Phi^{\dag} \right) \Big] \bigg\}}, \\
& = \min_{\Phi} \sqrt{1 - \frac{1}{n} \Real \bigg\{ \Tre \Big[
 \Phi^{\dag} \times \Trs \left( U V^{\dag} \right) \Big] \bigg\}}, \\
& = \min_{\Phi} \sqrt{1 - \frac{1}{n} \Real \Big[ \Tr \left( \Gamma
 \Phi^{\dag} \right) \Big]}.
\end{align}
\end{subequations}

Computing the singular value decomposition (SVD) of $\Gamma$ yields
$\Gamma = W \Sigma X^{\dag}$, where $W, X \in \Ue$ and $\Sigma =
\mathrm{diag} \left( \sigma_{1}, \ldots, \sigma_{n_{\rme}} \right) \in
\Mcal_{n_{\rme}} \left( \Rbb \right)$, with $\sigma_{i} \geq
\sigma_{i+1} \geq 0$, for all $i$. Thus,
\begin{equation}
\Tr \left( \Gamma \Phi^{\dag} \right) = \Tr \left( W \Sigma
 X^{\dag} \Phi^{\dag} \right) = \Tr \left( \Sigma X^{\dag}
 \Phi^{\dag} W \right) = \Tr \left( \Sigma Y \right),
\end{equation}
where $Y = X^{\dag} \Phi^{\dag} W \in \Ue$. Hence,
\begin{equation}
\Real \left[ \Tr \left( \Gamma \Phi^{\dag} \right) \right]
= \Real \left[ \Tr \left( \Sigma Y \right) \right] = \sum_{i
= 1}^{n_{\rme}} \sigma_{i} \Real \left( Y_{ii} \right),
\end{equation}
which, because $Y$ is unitary, achieves the maximum of $\Tr(\Sigma)$
when $Y = \Ie$. For a unitary matrix $Y$, as a result of orthonormality,
$\sum_{j} |y_{ij}|^{2} = 1$ for all $i$. Hence, $0 \leq |y_{ij}| \leq
1$ for all $i, j$, and in particular $0 \leq |y_{ii}| \leq 1$ for all
$i$. Since $\Sigma$ is diagonal and positive (with matrix elements
$\sigma_{i}$), $\Real \left[ \Tr \left( \Sigma Y \right) \right]
= \sum_{i} \sigma_{i} \Real \left( y_{ii} \right) \leq \sum_{i}
\sigma_{i} |y_{ii}| \leq \sum_{i} \sigma_{i} = \Tr \left( \Sigma
\right)$. Thus, the maximum of $\Real \left[ \Tr \left( \Sigma Y
\right) \right]$ is attained when $Y = \Ie$.

To show that $\Tr(\Sigma) = \| \Gamma \|_{\Tr}$, consider $\Sigma =
W^{\dag} \Gamma X = X^{\dag} \Gamma^{\dag} W$, which implies that
$\Sigma^{2} = X^{\dag} \Gamma^{\dag} \Gamma X$. Because $\Sigma$ is
diagonal, the matrix $\Sigma^{2}$ is also diagonal, and $X$ is the
unitary matrix which diagonalizes the positive matrix $\Gamma^{\dag}
\Gamma$. Therefore, it follows that $\sqrt{\Sigma^{2}} = \Sigma =
X^{\dag} \sqrt{\Gamma^{\dag} \Gamma} X$. Thus,
\begin{equation}
\max_{\Phi} \left\{ \Real \left[ \Tr \left( \Gamma
 \Phi^{\dag} \right) \right] \right\} = \Tr \left( \Sigma \right) =
 \| \Gamma \|_{\Tr},
\end{equation}
and (\ref{eq:HS-dist-b}) results.
\end{proof}

From this proof it follows that
\begin{equation}
\label{eq:convex}
\min_{\Phi} \Big\{ \| U - (\Is \otimes \Phi) V \|_{\rHS} : \Phi^{\dag}
\Phi = \Ie \Big\} = \min_{\Phi} \Big\{ \| U - (\Is \otimes \Phi) V
\|_{\rHS} : \Phi ^{\dag} \Phi \leq \Ie \Big\},
\end{equation}
and the minimum occurs on the boundary of the set $\Phi^{\dag} \Phi \leq
\Ie$, i.e., $\Phi = WX^{\dag} \in \Ue$, where $X$ and $W$ are the
unitary matrices from the SVD of $\Gamma$.

For the case where the target unitary transformation $\Vsys \in \Us$
is specified for a closed system, results of
section~\ref{ssec:closedsystem-target} can be used to obtain:
\begin{equation}
\label{eq:HS-dist-c} 
\tilde{\Delta}_{\rHS} ([U], [\Vsys \otimes \Ie]) = 
\left( 1 - \frac{1}{n} \| \Gamma \|_{\Tr} \right)^{1/2},
\hspace{0.5cm}
\Gamma = \Trs \left[ U (\Vsys^{\dag} \otimes \Ie ) \right].
\end{equation}

\subsection{Stability of the Hilbert-Schmidt norm}
\label{ssec:HS-stable}

Consider the stability \cite{Gilchrist05a} of the Hilbert-Schmidt norm,
as expressed in (\ref{eq:stable}):
\begin{equation}
\tilde{\Delta}_{\rHS} \big( [U] \otimes \Xi, [V] \otimes \Xi \big) =
\lambda_{n n_{\xi}} \min_{\Phi} \Big\{ \| [U - (\Is \otimes \Phi) V]
\otimes \Xi \|_{\rHS} : \Phi \in \Ue \Big\},
\end{equation}
where $\Xi$ is any unitary operator and $n_{\xi} = \dim \{ \Xi \}$. For
this norm, a tensor product of operators factors into the corresponding
product of the norms of those operators:
\begin{subequations}
\begin{align}
\tilde{\Delta}_{\rHS} \big( [U] \otimes \Xi, [V] \otimes \Xi \big)) & =
\lambda_{n} \min_{\Phi} \Big\{ \| [U - (\Is \otimes \Phi) V] \|_{\rHS}
\Big\} \times \| \Xi \|_{\rHS} / \sqrt{n_{\xi}}, \\
& = \lambda_{n} \min_{\Phi} \Big\{ \| [U - (\Is \otimes \Phi) V]
\|_{\rHS} \Big\}, \\
& = \tilde{\Delta}_{\rHS} ([U], [V]).
\end{align}
\end{subequations}
Therefore, with the appropriate normalization factor,
$\tilde{\Delta}_{\rHS}$ is stable. 

\subsection{Unitary operators that are exact tensor products}
\label{ssec:exact-tensor}

Suppose that $U$ and $V$, the actual and target operators, respectively,
are exact tensor products, i.e.,
\begin{equation}
\label{eq:tensor-product}
U = \Usys \otimes \Uenv \ \ \textrm{and} \ \ V = \Vsys \otimes \Venv,
\end{equation}
where $\Usys, \Vsys \in \Us$ and $\Uenv, \Venv \in \Ue$. Such a
situation would occur if, e.g., there is no system-environment
interaction ($H_{\rmint}= 0$) or the system and environment are decoupled
at a final time $t = t_{\mathrm{f}}$. Then, from (\ref{eq:gamma}),
\begin{equation}
\label{eq:gamma-dp}
\Gamma = \Tr \left( \Usys \Vsys^{\dag} \right) \Uenv \Venv^{\dag},
\end{equation}
and
\begin{equation}
\label{eq:gamma-dp-trace}
\| \Gamma \|_{\Tr} = \left| \Tr \left(
\Usys \Vsys^{\dag} \right) \right| \Tr \left( \sqrt{
\Venv \Uenv^{\dag} \Uenv \Venv^{\dag}} \right) = n_{\rme} 
\left| \Tr \left( \Usys \Vsys^{\dag} \right) \right|.
\end{equation}
The distance measure now becomes
\begin{equation}
\label{eq:dist-tensor}
\tilde{\Delta}_{\rHS} \big( [\Usys \otimes \Uenv], [\Vsys \otimes \Venv]
\big) = \left[ 1 - \frac{1}{n_{\rms}} \left| \Tr \left(
 \Usys \Vsys^{\dag} \right) \right| \right]^{1/2},
\end{equation}
which is, as one would expect, completely independent of $\Uenv$ and
$\Venv$. Indeed, as discussed in section~\ref{ssec:closedsystem-target},
$\Uenv$ and $\Venv$ can be incorporated into $\Phi$ in the quotient
metric in (\ref{eq:norm-dist}) and thus simply replaced by
$\Ie$. Therefore, the matrix $\Gamma$ in (\ref{eq:gamma-dp}) can be
replaced by $\Gamma = \Tr \left( \Usys \Vsys^{\dag} \right) \Ie$
[immediately leading to (\ref{eq:gamma-dp-trace})], and the distance
measure in (\ref{eq:dist-tensor}) can be expressed as
$\tilde{\Delta}_{\rHS} \big( [\Usys \otimes \Ie], [\Vsys \otimes \Ie]
\big)$.

The distance measure in (\ref{eq:dist-tensor}) is related to other
distance and fidelity measures used in previous works for optimal
control of unitary operations \cite{Palao02a, Palao03a, Sporl07a} and
analysis of the unitary control landscape \cite{Rabitz05a, Hsieh08a,
 Ho09a} in closed quantum systems. For control landscapes, the
analysis of critical points of $\tilde{\Delta}_{\rHS} \big( [\Usys
 \otimes \Ie], [\Vsys \otimes \Ie] \big)$ with respect to the control
field is especially important. 

Note that the distance in (\ref{eq:dist-tensor}) is zero if and only if
$\Usys$ and $\Vsys$ differ only by a global phase, i.e., $\Usys =
\exp(\rmi\phi) \Vsys$ for any real $\phi$. Therefore, this distance
can be considered as a phase-independent generalization of 
\begin{equation}
\label{eq:dist-standard-real}
\frac{1}{\sqrt{2n_{\rms}}}\| \Usys - \Vsys
 \|_{\rHS} = \left\{ 1 - \frac{1}{n_{\rms}} \Real
 \Big[ \Tr \left( \Usys \Vsys^{\dag}
 \right) \Big] \right\}^{1/2},
\end{equation}
a frequently used distance \cite{Ho09a} that is zero if and only if
$\Usys = \Vsys$.

\section{Computing the distance measure with the
maximum singular value}
\label{sec:singular}

Suppose a (possibly entangled) normalized pure state $|\psi\rangle$ of
the composite system $\Ccal$ is acted upon by a unitary operator $U$.
The output of this action, $U |\psi\rangle$, can be compared to the
output of the target action, $V |\psi\rangle$ (modulo $\Is \otimes
\Phi$). The output state error $\varepsilon$ is
\begin{equation}
\varepsilon := \left[ U - (\Is \otimes \Phi) V \right] |\psi\rangle.
\end{equation}
The maximum (over all normalized pure states) norm of the error is
\begin{equation}
\label{eq:2-norm}
\max_{|\psi\rangle} \varepsilon^{\dag} \varepsilon =
\| U - (\Is \otimes \Phi) V \|_{2}^{2},
\end{equation}
where $\| \cdot \|_{2}$ is the induced two-norm of a matrix
\cite{Horn90a}, i.e., the maximum singular value of the matrix
argument.\footnote{Equation~(\ref{eq:2-norm}) may also be expressed in
terms of the operator norm \cite{Horn90a}: $\displaystyle{\| U - (\Is
\otimes \Phi) V \|_{\infty} := \sup_{|\psi\rangle} \| \left[ U -
 (\Is \otimes \Phi) V \right] |\psi\rangle \|_{2}}$, where $\langle\psi|
\psi\rangle = 1$. In this context, $\| \cdot \|_{2}$ is the standard
vector two-norm.} Now consider the distance measure based on the induced
two-norm:
\begin{equation}
\label{eq:2-norm-dist}
\tilde{\Delta}_{2} ([U], [V]) = \lambda_{2} \min_{\Phi} \Big\{
\| U - (\Is \otimes \Phi) V \|_{2} \ \big| \ \Phi \in \Ue \Big\},
\end{equation}
where $\lambda_{2} = 2^{-1}$. Although an explicit solution for
$\tilde{\Delta}_{2} ([U], [V])$ is not presented in this work, we
establish the following bounds:
\begin{equation}
\label{eq:2-norm-bnds}
\lambda_{2} \| U - (\Is \otimes \underline{\Phi}) V \|_{2} \leq
\tilde{\Delta}_{2} ([U], [V]) \leq \lambda_{2} \| U - (\Is \otimes
\overline{\Phi}) V \|_{2},
\end{equation}
where $\underline{\Phi}$ is the solution to the following optimization
problem:
\begin{equation}
\label{eq:phi-over}
\underline{\Phi} := \arg \min_{\Phi} \Big\{ \| U - (\Is \otimes \Phi)
V\|_{2} : \Phi^{\dag} \Phi \leq \Ie \Big\}.
\end{equation}
Here, $\Phi^{\dag} \Phi \leq \Ie$ indicates that $\Phi^{\dag} \Phi$ is
a positive Hermitian matrix whose eigenvalues are at most 1. The
optimization in (\ref{eq:phi-over}) is convex for the same reasons as
those discussed in section~\ref{ssec:convexity}. However, there is no
guarantee that the resulting optimizer $\underline{\Phi}$ is
unitary. Since $\Phi$ in (\ref{eq:phi-over}) is less constrained than
in (\ref{eq:2-norm-dist}), it follows that the lower bound in
(\ref{eq:2-norm-bnds}) applies. The operator $\overline{\Phi}$ is obtained
from the SVD of $\underline{\Phi}$:
\begin{equation}
\label{eq:phi-hat}
\underline{\Phi} = WSX^{\dag} \ \to \ \overline{\Phi} = WX^{\dag}.
\end{equation}
Thus, $\overline{\Phi}$ is a unitary approximation to
$\underline{\Phi}$, and since it is not necessarily the optimal solution
to (\ref{eq:2-norm-dist}), the upper bound in (\ref{eq:2-norm-bnds})
follows. The lower and upper bounds in (\ref{eq:2-norm-bnds}) will be
close to each other if the singular values of $\underline{\Phi}$ are
close to unity, i.e., if $\underline{\Phi}$ is close to a unitary matrix.

For any matrix $M \in \Mcal_{n} (\Cbb)$, the induced
two-norm is bounded from above by the Hilbert-Schmidt norm
\cite{Horn90a}: $\| M \|_{2} \leq \| M \|_{\rHS}$. Comparing the
distance measure with these norms yields a similar relationship. Given
$U, V \in \Uc$, let $\Phi_{\rHS}$ be minimizer of $\tilde{\Delta}_{\rHS}
([U], [V])$. It follows that
\begin{align}
\lambda_{2}^{-1} \tilde{\Delta}_{2} ([U], [V]) & \leq \| U - \left( \Is
 \otimes \Phi_{\rHS} \right) V \|_{2} \nonumber \\
& \leq \|U - \left( \Is \otimes \Phi_{\rHS} \right) V \|_{\rHS} =
\lambda_{n}^{-1} \tilde{\Delta}_{\rHS} ([U], [V]),
\end{align}
i.e., $\tilde{\Delta}_{2} ([U], [V])$ is bounded from above by
$(n/2)^{1/2} \tilde{\Delta}_{\rHS} ([U], [V])$.

\section{Relating the distance measure to quantum channel
fidelity}
\label{sec:fidelity}

Quantum operations, for example, a free or controlled time evolution
(possibly in the presence of decoherence) or a transmission of
information (possibly via a noisy connection), are generally described
by quantum channels \cite{Nielsen00a, Reimpell05a}. In this section, we
consider a common situation in QC, where the target quantum channel is a
unitary transformation specified for a closed system (i.e., for
$\Scal$ alone). We denote this target unitary transformation as $\Vsys
\in \Us$ (cf.~sections~\ref{ssec:closedsystem-target} and
\ref{ssec:exact-tensor}), and the Hilbert-Schmidt norm distance
$\tilde{\Delta}_{\rHS} ([U], [\Vsys \otimes \Ie])$ corresponding to the
closed-system target is given by (\ref{eq:HS-dist-c}). We demonstrate
that the distance measure $\tilde{\Delta}_{\rHS} ([U], [\Vsys \otimes
\Ie])$ and a typical measure of \emph{quantum channel fidelity} both
depend on the matrix norm of $\Gamma = \Trs \left[ U (\Vsys^{\dag}
  \otimes \Ie ) \right]$ from (\ref{eq:HS-dist-c}).

There are several ways to define the fidelity of a quantum channel
with respect to the target unitary transformation \cite{Nielsen00a,
Uhlmann76a, Jozsa94a, Schumach96a, Fuchs99a, Kretsch04a, Gilchrist05a,
Lidar08a}. For example, the Uhlmann state fidelity, defined as
\cite{Uhlmann76a, Jozsa94a}
\begin{equation}
\label{eq:Uhlmann1}
\Fcal_{\rmu}(\rho_{1}, \rho_{2}) := \| \sqrt{\rho_{1}} \sqrt{\rho_{2}}
\|_{\Tr}^{2} = \left\{ \Tr \big[ \left( \sqrt{\rho_{1}} \rho_{2}
\sqrt{\rho_{1}} \right)^{1/2} \big] \right\}^{2} , 
\end{equation} 
where $\rho_{1}, \rho_{2} \in \Dcal(\Hcal)$, can be used to measure the
effect of quantum channels on states. The Kolmogorov distance between
two density matrices is \cite{Fuchs99a}
\begin{equation}
\label{eq:Kolmogorov}
D_{\rmk}(\rho_{1}, \rho_{2}) := \frac{1}{2} \| \rho_{1} - \rho_{2}
\|_{\Tr},
\end{equation}
which bounds the Uhlmann state fidelity \cite{Fuchs99a}:
\begin{equation}
\label{eq:dist-fidelity1}
\left[ 1 - D_{\rmk}(\rho_{1}, \rho_{2}) \right]^{2} \leq
\Fcal_{\rmu}(\rho_{1}, \rho_{2}) \leq 1 - D_{\rmk}^{2}(\rho_{1},
\rho_{2}).
\end{equation}
In this section, we develop a direct measure of fidelity for quantum
channels based on $\Fcal_{\rmu}$ in (\ref{eq:Uhlmann1}), and establish
bounds analogous to those in (\ref{eq:dist-fidelity1}) with the distance
measure $\tilde{\Delta}_{\rHS}$.

Let $\Kcal$ denote a completely positive and trace-preserving quantum
channel \cite{Nielsen00a}, mapping the state $\rho_{\rms}$ to
$\tilde{\rho}_{\rms}$, where $\rho_{\rms}, \, \tilde{\rho}_{\rms} \in
\Dcal (\Hs)$, with the following operator-sum representation (also known
as the Kraus map \cite{Kraus83a}):
\begin{subequations}
\label{eq:OSR}
\begin{gather}
\tilde{\rho}_{\rms} = \Kcal[ \rho_{\rms} ] := 
\sum_{i} K_{i} \rho_{\rms} K_{i}^{\dag}, \\
\intertext{such that}
\sum_{i} K_{i}^{\dag} K_{i} = \Is, \ \ \textrm{where} \ K_{i} \in
\Mcal_{n_{\rms}} (\Cbb).
\end{gather}
\end{subequations}
The target quantum channel is the unitary transformation $V_{\rms}$ that
maps $\rho_{\rms}$ to $V_{\rms}\rho_{\rms}V_{\rms}^{\dag}$.

We can compare the actual and target quantum channels, $\Kcal$ and
$V_{\rms}$, respectively, using the Uhlmann fidelity of
(\ref{eq:Uhlmann1}):
\begin{equation}
\label{eq:Uhlmann2}
\Fcal_{\rmu} \left( V_{\rms} \rho_{\rms} V_{\rms}^{\dag},
\tilde{\rho}_{\rms} \right) = \left\{ \Tr \left[ \left( V_{\rms}
\sqrt{\rho_{\rms}} V_{\rms}^{\dag} \tilde{\rho}_{\rms} V_{\rms}
\sqrt{\rho_{\rms}} V_{\rms}^{\dag} \right)^{1/2} \right] \right\}^{2} .
\end{equation}
When the input state $\rho_{\rms}$ is pure, it can be expressed as
$\rho_{\rms} = |\psi_{\rms}\rangle \langle\psi_{\rms}|$; then
(\ref{eq:Uhlmann2}) yields
\begin{equation}
\label{eq:Uhlmann3}
\Fcal_{\rmu} \left(
V_{\rms} |\psi_{\rms}\rangle \langle\psi_{\rms}| V_{\rms}^{\dag},
\tilde{\rho}_{\rms} \right) 
= \langle\psi_{\rms}| V_{\rms}^{\dag} \tilde{\rho}_{\rms} 
V_{\rms} |\psi_{\rms}\rangle
= \sum_{i} \left|
\langle \psi_{\rms} | \Vsys^{\dag} K_{i} | \psi_{\rms} \rangle
\right|^{2} .
\end{equation}
To evaluate the proximity of the quantum channel $\Kcal$ to the
target transformation $\Vsys$, it is convenient to use a measure that is
independent of the input state. One such measure is the \emph{minimum
pure-state fidelity} $\Fcal_{\mathrm{p}}(\Kcal, \Vsys)$
\cite{Nielsen00a}, which is obtained by minimizing the form
(\ref{eq:Uhlmann3}) of the Uhlmann fidelity over the set of all pure
input states:
\begin{equation}
\label{eq:fidelity-MPS}
\Fcal_{\mathrm{p}}(\Kcal, \Vsys) :=
\min_{|\psi_{\rms}\rangle} 
\Fcal_{\rmu} \left(
V_{\rms} |\psi_{\rms}\rangle \langle\psi_{\rms}| V_{\rms}^{\dag},
\tilde{\rho}_{\rms} \right) 
= \min_{|\psi_{\rms}\rangle} \sum_{i} \left|
\langle \psi_{\rms} | \Vsys^{\dag} K_{i} | \psi_{\rms} \rangle \right|^{2}.
\end{equation}
Calculating $\Fcal_{\mathrm{p}}(\Kcal, \Vsys)$ is not easy
in general because the minimization in (\ref{eq:fidelity-MPS}) is not
convex. However, this problem can be related to a convex optimization,
producing the following lower bound:
\begin{subequations}
\label{eq:fidelity-bnd}
\begin{gather}
\underline{\Fcal}(\Kcal, \Vsys) \leq
\Fcal_{\mathrm{p}} (\Kcal, \Vsys), \\
\intertext{where}
\underline{\Fcal}(\Kcal, \Vsys) := \min_{\rho_{\rms}}
\sum_{i} \left| \Tr \left( \Vsys^{\dag} K_{i} \rho_{\rms} \right)
\right|^{2}.
\end{gather}
\end{subequations}
Both fidelities are contained in the interval $[0,1]$ and equal to one
if and only if $\Kcal \left[ \rho_{\rms} \right] = \Vsys
\rho_{\rms} \Vsys^{\dag}$ for all density matrices $\rho_{\rms} \in
\Dcal (\Hs)$, i.e., if $K_{i} = \alpha_{i} \Vsys$ for all $i$ and
$\sum_{i} \left| \alpha_{i} \right|^{2} = 1$. For a given $\Kcal$,
$\underline{\Fcal}(\Kcal, \Vsys)$ is found via a convex
optimization over all states $\rho_{\rms} \in \Dcal (\Hs)$, and
hence, it can be efficiently obtained numerically. It is shown in
\cite{Nielsen00a} that $\underline{\Fcal}(\Kcal, \Vsys)$ and
$\Fcal_{\mathrm{p}} (\Kcal, \Vsys)$ have the same pure-state
minimum, but since the set of all density matrices is convex, while the
set of all pure states is not, finding
$\underline{\Fcal}(\Kcal, \Vsys)$ is an easier optimization
problem.

If, instead of minimizing over all possible density matrices,
$\rho_{\rms}$ in $\underline{\Fcal}(\Kcal, \Vsys)$ is
simply the so-called \emph{maximally-mixed} system state $\rho_{\rmm,
 \rms} := \Is/n_{\rms}$, this results in another variant of quantum
channel fidelity, $\Fcal_{\rmm, \rms}(\Kcal, \Vsys)$ \cite{Nielsen00a,
Kretsch04a, Reimpell05a}: 
\begin{equation}
\label{eq:channel-fidelity}
\Fcal_{\rmm, \rms}(\Kcal, \Vsys) := \frac{1}{n_{\rms}^{2}}
\sum_{i} \left| \Tr \left( \Vsys^{\dag} K_{i} \right) \right|^{2}.
\end{equation}
This fidelity also evaluates the proximity of the quantum channel
$\Kcal$ to the target transformation $\Vsys$ \cite{Nielsen00a},
but without state optimization. Interestingly, $\rho_{\rmm, \rms}$ has
the unique property of being the density matrix with the ``shortest''
maximum distance, under the Hilbert-Schmidt norm, to any other density
matrix:
\begin{subequations}
\begin{gather}
\rho_{\rmm, \rms} = \arg \min_{\rho_{2}} \left( \max_{\rho_{1}} \|
\rho_{1} - \rho_{2} \|_{\rHS}^{2} \right), \ \ \forall \rho_{1},
\rho_{2} \in \Dcal(\Hs), \\ 
\intertext{where}
\min_{\rho_{2}} \left( \max_{\rho_{1}} \| \rho_{1} - \rho_{2}
\|_{\rHS}^{2} \right) = \frac{n_{\rms}-1}{n_{\rms}}.
\end{gather}
\end{subequations}
This result suggests that $\rho_{\rmm, \rms}$ is ``centrally located''
within the space of density matrices in the Hilbert-Schmidt geometry.

\begin{proof}
See \ref{app:rho-min-max}.
\end{proof}

Relating $U$ (the composite-system evolution operator) to $\Kcal$ (the
quantum channel for the system $\Scal$) requires the specification of an
initial state of the composite system $\Ccal$. For simplicity, assume
that this initial state is an uncorrelated tensor-product state:
\begin{equation}
\rho = \rho_{\rms} \otimes \rho_{\rme} = \rho_{\rms} \otimes \sum_{\nu =
 1}^{n_{\rme}} \zeta_{\nu} |\nu \rangle \langle \nu|,
\end{equation}
where, as before, $\rho \in \Dcal (\Hc)$, $\rho \geq 0$, and $\Tr \left(
  \rho \right) = 1$. The reduced dynamics of the system $\Scal$ can be
represented by the mapping in (\ref{eq:OSR}), where the Kraus operators
$K_{\nu \nu'}$ are
\begin{equation}
\label{eq:Kraus-ops}
K_{\nu \nu'} := \sqrt{\zeta_{\nu'}} \Tre \left[ (\Is \otimes
 |\nu'\rangle \langle\nu|) U \right] = \sqrt{\zeta_{\nu'}} \sum_{i,
 i' = 1}^{n_{\rms}} U_{\substack{i i' \\ \nu \nu'}} |i\rangle \langle i'|,
\end{equation}
with the unitary operator $U$ expanded as
\begin{equation}
U = \sum_{i, i' = 1}^{n_{\rms}} \sum_{\nu, \nu' = 1}^{n_{\rme}}
U_{\substack{i i' \\ \nu \nu'}} |i\rangle \langle i'| \otimes |\nu
\rangle \langle \nu'|.
\end{equation}

There exist infinitely many different sets of Kraus operators, $\{ K_{j}
\}$, that represent the same map $\Kcal$ (i.e., they evolve the system
state $\rho_{\rms}$ in exactly the same way) \cite{Kraus83a}. This is
related to the equivalence relation defined in section~\ref{sec:equiv},
namely that $U$ and $(\Is \otimes \Phi)U$ produce the same system
evolution and hence, the same Kraus map. Moreover, any Kraus map for a
$k$-level system can be represented by a set of at most $k^{2}$ Kraus
operators.

The mapping from $U$ and $\rho$ to $\Kcal$ described above illustrates
the dependence of the quantum channel $\Kcal$ (and hence, of the
corresponding fidelities) on the initial state of the environment. Note
that the exact treatment of the dynamics requires the propagation of the
evolution operator for the composite system $\Ccal$, which is also
required for computing $\tilde{\Delta} ([U], [\Vsys \otimes
\Ie])$. Although approximations may simplify the calculation of $\Kcal$,
they may eliminate certain physical processes, e.g., the Markovian
approximation limits the memory of the environment, preventing possible
coherence revivals. However, even when no approximations are made,
measures of fidelity for a quantum channel $\Kcal$ require (for the
mapping from $U$ to $\Kcal$) the specification of the environment's
initial state and thus are less general than the distance measure
$\tilde{\Delta} ([U], [\Vsys \otimes \Ie])$.

With (\ref{eq:Kraus-ops}), the fidelity $\Fcal_{\rmm,
\rms}(\Kcal, \Vsys)$ in (\ref{eq:channel-fidelity}) becomes
\begin{equation}
\Fcal_{\rmm, \rms}(\Kcal, \Vsys) = \frac{1}{n_{\rms}^{2}}
\sum_{\nu, \nu' = 1}^{n_{\rme}} \zeta_{\nu'} 
\left| \Gamma_{\nu \nu'} \right|^{2} = \frac{1}{n_{\rms}^{2}} 
\| \Gamma \sqrt{\rho_{\rme}} \|_{\rHS}^{2},
\end{equation}
where $\rho_{\rme} \in \Dcal (\He)$ is the initial state of the
environment, $\rho_{\rme} = \sum_{\nu = 1}^{n_{\rme}} \zeta_{\nu} |\nu
\rangle \langle \nu|$. Also, we see that both the distance measure
$\tilde{\Delta}_{\rHS} \big( [U], [\Vsys \otimes \Ie] \big)$ and
quantum channel fidelity $\Fcal_{\rmm, \rms}(\Kcal, \Vsys)$ depend on
the matrix $\Gamma$ defined in (\ref{eq:HS-dist-c}). Specifically,
increasing the norm of $\Gamma$ will increase fidelity and decrease
distance. As a concluding example, suppose that $\rho_{\rme}$ is the
maximally mixed environment state, i.e., $\rho_{\rmm, \rme} :=
\Ie/n_{\rme}$. In this case, the quantum channel fidelity $\Fcal_{\rmm,
\rms}(\Kcal, \Vsys)$ becomes
\begin{equation}
\label{eq:fidelity-mixed}
\Fcal_{\mathrm{m,c}}(\Kcal, \Vsys) := \frac{1}{n_{\rms} n}
\sum_{\nu, \nu' = 1}^{n_{\rme}} \left|
 \Gamma_{\nu \nu'} \right|^{2} = \frac{1}{n_{\rms} n} \| \Gamma
\|_{\rHS}^{2},
\end{equation}
further emphasizing the dependence of fidelity on the norm of $\Gamma$.

We show that the inequality relation between the distance and fidelity
of quantum states presented in (\ref{eq:dist-fidelity1}) \cite{Fuchs99a}
also applies to the distance and fidelity of quantum operations given by
$\tilde{\Delta}_{\rHS} ([U], [\Vsys \otimes \Ie])$ and
$\Fcal_{\mathrm{m,c}}(\Kcal, \Vsys)$, respectively. Specifically, we
obtain the following lower and upper bounds of
$\Fcal_{\mathrm{m,c}}(\Kcal, \Vsys)$:
\begin{equation}
\label{eq:dist-fidelity2}
\left[ 1 - \tilde{\Delta}_{\rHS} ([U], [\Vsys \otimes \Ie]) \right]^{2}
\leq \Fcal_{\mathrm{m,c}}(\Kcal, \Vsys) \leq 1 -
\tilde{\Delta}^{2}_{\rHS} ([U], [\Vsys \otimes \Ie]).
\end{equation}

\begin{proof}
See \ref{app:fidelity}.
\end{proof}

\section{Applications: Optimal control of quantum gates}
\label{sec:control}

\subsection{Model open system}

In this section, we illustrate the use of the distance measure
$\tilde{\Delta}_{\rHS} ([U], [\Vsys \otimes \Ie])$ in the optimal
control of one- and two-qubit gates relevant for QC. Specifically, this
measure is used in numerical calculations to evaluate the distance
between quantum operations generated in the framework of optimal control
theory \cite{Rabitz00b, Balint08a} for a system coupled to a decohering
environment \cite{Grace07a, Grace07b}, and a target unitary
transformation. We use a model of interacting two-level particles, which
are divided into a QIP, composed of $q$ qubits, and an environment,
composed of $e$ two-level particles. The qubits are directly coupled to
a time-dependent external control field, while the environment is not
directly controlled and is thereby managed only through its interaction
with the qubits. The Hamiltonian for the composite system $\Ccal$,
abbreviated as $H = H_{0} + H_{c(t)} + H_{\rmint}$,\footnote{Compare
this expression to the abbreviated Hamiltonian in
(\ref{eq:Hamiltonian1}): $H_{0} + H_{c(t)} = H_{\rms} + H_{\rme}$.} has
the explicit form
\begin{equation}
\label{eq:Hamiltonian2}
H = \sum_{i = 1}^{q + e} \omega_{i} S_{iz} 
- \sum_{i = 1}^{q} \mu_{i} C(t) S_{ix}
- \sum_{j = i+1}^{q + e} \sum_{i = 1}^{q + e - 1} \gamma_{ij}
\mathbf{S}_{i} \cdot \mathbf{S}_{j}.
\end{equation}
Here, $\mathbf{S}_{i} = \left( S_{ix}, S_{iy}, S_{iz} \right)$ is the
spin operator for the $i$th particle ($\mathbf{S}_{i} = \frac{1}{2}
\bm{\upsigma}_{i}$, in terms of the Pauli matrices), $H_{0}$ is the
sum over the free Hamiltonians $\omega_{i} S_{iz}$ for all $q + e$
particles ($\omega_{i}$ is the transition angular frequency for the
$i$th particle), $H_{c(t)}$ specifies the coupling between the $q$
qubits and the time-dependent control field $C(t)$ ($\mu_{i}$ are the
corresponding dipole moments), and $H_{\rmint}$ represents the
Heisenberg exchange interaction between the particles ($\gamma_{ij}$
is the coupling parameter for the $i$th and $j$th
particles). Hamiltonians of the form in (\ref{eq:Hamiltonian2}) are
often referred to as \emph{effective} or \emph{spin} Hamiltonians
\cite{Tsaparlis77a}.

The evolution of the composite system $\Ccal$ is calculated in
an exact quantum-mechanical manner, by propagating the Schr\"{o}dinger
equation (\ref{eq:Schrodinger}) for the Hamiltonian in
(\ref{eq:Hamiltonian2}), without either approximating the dynamics by
a master equation or using a perturbative analysis based on a weak
coupling assumption. This calculation produces the final-time
evolution operator $U_{\tf} \in \Uc$ (an $n \times n$ unitary matrix),
which is compared to the target gate operation $\Vsys \in \Us$ (an
$n_{\rms} \times n_{\rms}$ unitary matrix) via $\tilde{\Delta}_{\rHS}
\big( [U_{\tf}], [\Vsys \otimes \Ie] \big)$. Note that $n_{\rms} =
2^{q}$, $n_{\rme} = 2^{e}$, and $n = 2^{(q + e)}$. 

For our examples, we consider two different composite systems. In the
first example, one qubit is coupled to a two-particle environment ($q
= 1$ and $e = 2$), which can be modeled by a linear chain of particles
with the qubit $q_{1}$ at the center, equally coupled to both
environment particles $e_{2}$ and $e_{3}$:
\begin{equation}
e_{2} \stackrel{\gamma_{12}}{\longleftrightarrow} q_{1}
\stackrel{\gamma_{13}}{\longleftrightarrow} e_{3} \ ,
\end{equation}
where $\gamma_{12} = \gamma_{13} = \gamma$. Assuming only
nearest-neighbor coupling in this configuration, we set $\gamma_{23} =
0$. In the second example, two qubits are equally coupled to a
one-particle environment ($q = 2$ and $e = 1$). This is modeled as a
triangular cluster: 
\begin{equation}
\begin{array}{ccc}
& e_{3} & \\ 
\stackrel{\gamma_{13}}{} \swarrow \! \! \! \! \! \! \nearrow & &
\nwarrow \! \! \! \! \! \! \searrow \stackrel{\gamma_{23}}{} \\ 
q_{1} \ \ & \stackrel{\gamma_{12}}{\longleftrightarrow} & \ \ q_{2}
\end{array} \ ,
\end{equation}
where the two qubits are denoted as $q_{1}$ and $q_{2}$, the
environment particle as $e_{3}$, and $\gamma_{13} = \gamma_{23} =
\gamma$. Other system-environment configurations and corresponding
optimal control results, using $\tilde{\Delta}_{\rHS} \big( [U_{\tf}],
[\Vsys \otimes \Ie] \big)$, are presented in \cite{Grace07a, Grace07b}. 

Parameters of the composite system are selected to ensure complex
dynamics and strong decoherence: values of $\gamma/\omega$ are up to
0.02 and frequencies $\omega_{i}$ are close, but not equal, to enhance
the system-environment interaction (see \cite{Grace07a} for more
details). In addition to setting $\hbar = 1$, we introduce a natural
system of units by setting $\mu_{i} = 1$ for all $i$ and the qubit
frequency $\omega_{1} = 1$, implying that one period of free evolution
is $2 \pi$.\footnote{For one-qubit coupled to a two-particle
 environment, the frequencies of the environment particles are:
 $\omega_{2} \approx 0.99841$, $\omega_{3} \approx 1.00159$. For two
 qubits coupled to a one-particle environment, the frequencies of
 $q_{2}$ and $e_{3}$ are: $\omega_{2} \approx 1.09159$ and $\omega_{3}
 \approx 0.99841$, respectively.}

\subsection{Optimal control algorithm}
\label{ssec:algorithm}

The ultimate control goal is to decouple the system $\Scal$ from
the environment $\Ecal$ at a time $t = t_{\mathrm{f}}$ and
simultaneously produce the target unitary operation $\Vsys \in \Us$
for $\Scal$. Target quantum gates for the one- and two-qubit
systems are the Hadamard gate $H_{\mathrm{g}}$ and controlled-not gate
CNOT, respectively (both of these gates are elements of a universal set
of logical operations for QC \cite{Nielsen00a}):
\begin{equation}
H_{\mathrm{g}} := \frac{1}{\sqrt{2}} \left( \begin{array}{cr} 
1 & 1 \\ 
1 & -1 
\end{array} \right) \ \ \textrm{and} \ \
\mathrm{CNOT} := \left( \begin{array}{cccc} 
1 & 0 & 0 & 0 \\ 
0 & 1 & 0 & 0 \\ 
0 & 0 & 0 & 1 \\ 
0 & 0 & 1 & 0 
\end{array} \right). 
\end{equation}

A hybrid optimization method incorporating genetic and gradient
algorithms \cite{Blum08a} is employed to minimize the distance measure
$\tilde{\Delta}_{\rHS} \big( [U_{\tf}], [\Vsys \otimes \Ie] \big)$
with respect to the control field $C(t)$. After an initial
optimization with a genetic algorithm, we remove the constraints
imposed by the parameterized form of the control field to provide the
potential for more flexible and effective control. An optimal control
field is then found by minimizing the objective functional $\Jsys$,
given by
\begin{equation}
\label{eq:functional}
\Jsys(C) := \tilde{\Delta}_{\rHS} \big( [U_{\tf} (C)], 
[\Vsys \otimes \Ie] \big)
+ \frac{\alpha}{2} \| C \|_{\Kbb_{\tf}}^{2},
\end{equation}
using a gradient algorithm (see \cite{Palao03a, Grace07a} for details
of the gradient-based optimization). Here, $\Kbb_{\tf}$ denotes
a particular closed subspace of $L^{2} \left( [0,\tf]; \Rbb \right)$
(see \ref{app:objective} for details), and for the time-dependent
Hamiltonian in (\ref{eq:Hamiltonian2}), $U_{\tf}: \Kbb_{\tf} \to \Uc$
denotes the map, defined implicitly through the Schr\"{o}dinger equation
(\ref{eq:Schrodinger}), that takes a control field $C(t) \in \Kbb_{\tf}$
to the unitary time-evolution operator $U_{\tf} \in \Uc$. Thus,
$\Kbb_{\tf}$ is a Hilbert space of admissible controls, on which $U(t;
C)$ exists for all $t \in [0, \tf]$ and all $C \in \Kbb_{\tf}$
\cite{Jurdjevic72a}. As such, $\Jsys : \Kbb_{\tf} \to \Rbb$ is the
``dynamical'' version of the distance measure $\tilde{\Delta}_{\rHS}$,
with an additional cost on the control field fluence, where $\alpha > 0$
is the weight parameter for this cost. The gradient of $\Jsys$ is
derived in \ref{app:objective}.

\subsection{Results}

Because $0 \leq \tilde{\Delta}_{\rHS} \leq 1$ in general, it is
convenient to define gate fidelities based on the respective lower and
upper bounds of $\Fcal_{\mathrm{m,c}}(\Kcal, \Vsys)$ in
(\ref{eq:dist-fidelity2}):
\begin{subequations}
\label{eq:fidelity}
\begin{gather}
\Fcal_{\rHS}^{\rml} \big( [U_{\tf}], [\Vsys \otimes \Ie]
\big) := \left[ 1 - \tilde{\Delta}_{\rHS} ([U_{\tf}], [\Vsys \otimes
\Ie]) \right]^{2}, \\
\intertext{and} 
\Fcal_{\rHS}^{\rmu} \big( [U_{\tf}], [\Vsys \otimes \Ie]
\big) := 1 - \tilde{\Delta}_{\rHS}^{2} \big( [U_{\tf}], [\Vsys \otimes
\Ie] \big).
\end{gather}
\end{subequations}
Both fidelities are independent of the initial state and are evaluated
directly from the evolution operator $U_{\tf}$ of the composite
system. These fidelities, computed for the one-qubit Hadamard gate,
optimally controlled in the presence of a two-particle environment,
are presented in figure~\ref{fig:fidelity} for various values of the
qubit-environment coupling parameter $\gamma$. For comparison, a
fidelity based on the upper bound of the induced two-norm distance [as
defined in (\ref{eq:2-norm-bnds})] is also presented:
\begin{equation}
\Fcal_{2} \big( [U_{\tf}], [\Vsys \otimes \Ie] \big) := 
\left[ 1 - \lambda_{2} \| U_{\tf} - (\Vsys \otimes \overline{\Phi}) \|_{2} 
\right]^{2},
\end{equation}
where $\overline{\Phi}$ is defined in
(\ref{eq:2-norm-bnds})--(\ref{eq:phi-hat}). This fidelity was cast as a
semidefinite program \cite{Boyd04a} and calculated from the final-time
evolution operator $U_{\tf}$ using the SDPT3 solver \cite{Toh99a} and
the YALMIP toolbox \cite{Lofberg04a} in MATLAB. Since the evolution
operators $U_{\tf}$ produced via optimal control in this example are
very close to the target equivalence class, the upper and lower bounds
on the induced two-norm distance are nearly indistinguishable, although
this is not true in general.\footnote{For randomly selected unitary
operators in $\Uc$, the bounds in (\ref{eq:2-norm-bnds}) are well
separated and the minimizer $\underline{\Phi}$ defined in
(\ref{eq:phi-over}) is not unitary.}

\begin{figure}[ht]
\label{fig:fidelity}
\epsfxsize=0.8\textwidth \centerline{\epsffile{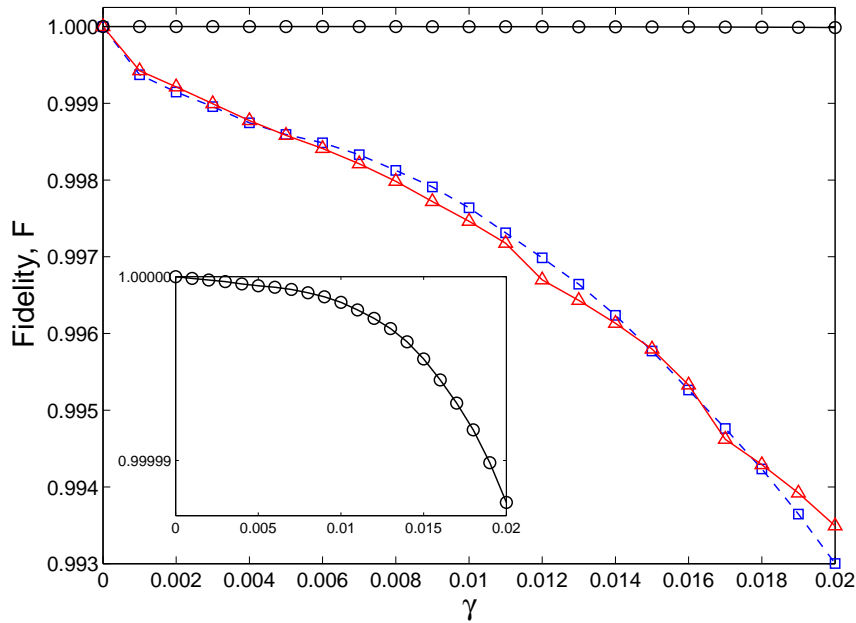}}
\caption{Gate fidelities $\Fcal_{\rHS}^{\rmu}$ (circles,
 also shown at a different scale in the inset),
 $\Fcal_{\rHS}^{\rml}$ (squares), and $\Fcal_{2}$
 (triangles) versus the coupling parameter $\gamma$, for the
 optimally controlled one-qubit Hadamard gate $H_{\mathrm{g}}$. The
 qubit is equally coupled to a two-particle environment ($q = 1$, $e
 = 2$). Values of $\gamma$ range from 0 to 0.02 in increments of
 0.001.}
\end{figure}

For the one-qubit system with no coupling to the environment ($\gamma
= 0$), the optimal control field produces the Hadamard gate with the
following values of the distance and fidelities:
$\tilde{\Delta}_{\rHS} \approx 10^{-6}$,
$\Fcal_{\rHS}^{\rmu} \approx 1 - 10^{-12}$,
$\Fcal_{\rHS}^{\rml} \approx 1 - 2 \times 10^{-6}$. For
the maximum coupling to the environment ($\gamma = 0.02$), we obtain for
the optimally controlled Hadamard gate: $\tilde{\Delta}_{\rHS} \approx
0.0035$, $\Fcal_{\rHS}^{\rmu} \approx 1 - 10^{-5}$,
$\Fcal_{\rHS}^{\rml} \approx 0.993$. For the two-qubit system with
no coupling to the environment ($\gamma = 0$), the optimal control field
produces the CNOT gate with the following values of the distance and
fidelities: $\tilde{\Delta}_{\rHS} \approx 10^{-4}$,
$\Fcal_{\rHS}^{\rmu} \approx 1 - 10^{-8}$,
$\Fcal_{\rHS}^{\rml} \approx 0.9998$. For the maximum coupling to
the environment ($\gamma = 0.01$), we obtain for the optimally
controlled CNOT gate: $\tilde{\Delta}_{\rHS} \approx 0.02$,
$\Fcal_{\rHS}^{\rmu} \approx 0.9996$, $\Fcal_{\rHS}^{\rml}
\approx 0.96$.

More significant than the actual fidelity values obtained in this
example is the demonstrated ability of the distance measure
$\tilde{\Delta}_{\rHS} \big( [U_{\tf}], [\Vsys \otimes \Ie] \big)$ to
quantitatively compare the evolution of a realistic open system to a
target unitary transformation specified for an ideal closed
system. Thus, the use of the distance $\tilde{\Delta}_{\rHS}
\big( [U_{\tf}], [\Vsys \otimes \Ie] \big)$ allows for the direct
optimization of quantum operations in the presence of an environment,
without specifying any part of the composite-system state.

\section{Conclusions}
\label{sec:conclusion}

We presented the novel and useful distance measure $\tilde{\Delta}
([U], [V])$ for the quantitative comparison of unitary quantum
operations acting on a composite quantum system, where the effect on
only one sub-system (e.g., a QIP) is important, while the effect on
the rest of the system (e.g., an environment) can be arbitrary. In
practically important situations where the target operation is
specified only for the sub-system of interest, the corresponding
measure $\tilde{\Delta} ([U], [\Vsys \otimes \Ie])$ can evaluate the
distance between the evolution of the entire composite system and the
target transformation (i.e., effectively compare unitary operations of
different dimensions). This capability is especially desirable for
measuring the distance between the actual controlled quantum process
in a realistic open system and the target unitary transformation in
its ideal closed counterpart. The measure developed in this paper does
not require the specification of the initial state of any component of
the system, is straightforward to calculate with the Hilbert-Schmidt
norm, and can be used to define relevant state-independent fidelity
measures. A fidelity measure independent of the initial state is
extremely valuable for optimal control of quantum gates, which is a
crucially important task in realistic QC.

The quotient metric developed here is also applicable to other
proposals for design and control of quantum operations. For example,
it can be used for evaluating the fidelity of quantum gates in
situations where multilevel encoding (MLE) of quantum logical states
is employed (i.e., where each logical basis state of a qubit is
encoded by multiple physical levels of a quantum system)
\cite{Grace06a}. Target unitary operations for a QIP equipped with
MLE have a tensor-product structure similar to that in
(\ref{eq:tensor-product}), and the \emph{logical equivalence} can exist
between different physical operations (i.e., different physical
processes can result in the same quantum logical operation on the
QIP). However, the environment is not explicitly included in the MLE
formalism.

Future research arising from the present work will involve the analysis
of quantum control landscapes \cite{Rabitz05a, Hsieh08a, Ho09a,
Rothman05a, Rothman05b, Rothman06a, Dominy08a, Brif09a} for objective
functionals based on various forms of the distance measure
$\tilde{\Delta}([U],[V])$, in particular, on the Hilbert-Schmidt norm
distance $\tilde{\Delta}_{\rHS} \big( [U_{\tf} (C)], [\Vsys \otimes \Ie]
\big)$. It would be interesting to investigate the structure of these
control landscapes and systematically search for robust control
solutions that incorporate requirements on realistic physical
resources. This search can be performed using a method called
diffeomorphic modulation under observable-response-preserving homotopy
(D-MORPH) \cite{Rothman05a, Rothman05b, Rothman06a, Dominy08a}. This
method provides an accurate numerical tool for finding optimal
controls. In particular, D-MORPH applied to closed quantum systems is
able to identify optimal controls generating a target unitary
transformation up to machine precision \cite{Dominy08a}. Extending the
applicability of this approach to open quantum systems is an important
goal, currently under development. Implementing efficient numerical
algorithms to calculate unitary time-evolution operators for larger
composite systems is essential for practical optimal control
simulations. Recent work has shown that such unitary maps can be
efficiently constructed from appropriately designed state-to-state maps
\cite{Merkel09a}, utilizing the favorable state-control landscape
structure \cite{Rabitz04a, Rabitz05a}. The possibility of expanding the
method developed in \cite{Merkel09a} to the optimal control of open
quantum systems is being explored.

\ack 
This work was supported by the NSF and ARO. JD~acknowledges support from
the Program in Plasma Science and Technology at Princeton. Sandia is a
multipurpose laboratory operated by Sandia Corporation, a
Lockheed-Martin Company, for the United States Department of Energy
under contract DE-AC04-94AL85000.

\appendix
\section{Equivalence of unitary operations}
\label{app:equiv}

The theorem from section~\ref{sec:equiv} is repeated here for
convenience.

\setcounter{theorem}{0}
\begin{theorem}
Let $U, V \in \Uc$. Then $\Tre \left( U \rho U^{\dag} \right) = \Tre
\left( V \rho V^{\dag} \right)$ for all density matrices $\rho \in
\Dcal(\Hc)$ if and only if $U = (\Is \otimes \Phi) V$ for some
$\Phi \in \Ue$.
\end{theorem}

\begin{proof} First, if $U = (\Is \otimes \Phi)V$, then $\Tre \left(
U\rho U^{\dag} \right) = \Tre \big[ (\Is \otimes \Phi) V\rho V^{\dag}
\left( \Is \otimes \Phi^{\dag} \right) \big] = \Tre \left( V\rho
V^{\dag} \right)$, so the ``if'' implication is proved. Next, to show
the ``only if'' implication, suppose that $\Tre \left( U\rho U^{\dag}
\right) = \Tre \left( V \rho V^{\dag} \right)$ for all density matrices
$\rho$. By complex linearity, $\Tre \left( U A U^{\dag} \right) = \Tre
\left( V A V^{\dag} \right)$ for all complex matrices $A \in
\Mcal_{n} (\Cbb)$. Therefore, $\langle B, \Tre \left( U A U^{\dag}
\right) \rangle = \langle B, \Tre \left( V A V^{\dag} \right) \rangle$
for all $A \in \Mcal_{n} (\Cbb)$ and all $B \in \Mcal_{n_{\rms}}
(\Cbb)$, where
\begin{equation}
\label{eq:HS-in-prod}
\langle B_{1}, B_{2} \rangle = \Real \left[ \Tr \left(
 B_{1}^{\dag} B_{2} \right) \right], \ \ \forall B_{1}, B_{2} \in
\Mcal_{n_{\rms}} (\Cbb),
\end{equation}
is the real Hilbert-Schmidt inner product on $\Mcal_{n_{\rms}}
(\Cbb)$ [analogously for $\Mcal_{n} (\Cbb)$]. If
\begin{subequations}
\begin{align}
\label{eqa:adj-S}
& \langle B, \Tre \left( U A U^{\dag} \right) \rangle =
\langle B, \Tre \left( V A V^{\dag} \right) \rangle, \ \ \forall A \in
\Mcal_{n} (\Cbb) \ \ \textrm{and} \ \ B \in \Mcal_{n_{\rms}} (\Cbb), \\
\intertext{then}
\label{eq:adj-C}
& \langle \Tre^{*}(B), U A U^{\dag} \rangle = \langle \Tre^{*}(B), V A
V^{\dag} \rangle, \ \ \forall A \in \Mcal_{n} (\Cbb) \ \ \textrm{and} \
\ B \in \Mcal_{n_{\rms}} (\Cbb),
\end{align}
\end{subequations}
where $\Tre^{*}( \cdot )$ is the operator adjoint of
$\Tre(\cdot)$. Starting from $\langle B, \Tre(A) \rangle = \langle
\Tre^{*}(B), A \rangle$, for all $B \in \Mcal_{n_{\rms}} (\Cbb)$ and $A
\in \Mcal_{n} (\Cbb)$, we show that $\Tre^{*}(B) = B \otimes \Ie$:
\begin{subequations}
\begin{align}
\langle B, \Tre(A) \rangle & = \sum_{i, k = 1}^{n_{\rms}}
\sum_{j = 1}^{n} \left( B^{\dag} \right)_{ki} A_{ijkj}, \\ 
& = \sum_{i, k = 1}^{n_{\rms}} \sum_{j, \ell = 1}^{n} \left( B^{\dag}
 \otimes \Ie \right)_{k\ell ij} A_{ijk\ell}, \\
& = \langle \Tre^{*}(B), A \rangle \ \Rightarrow \ \Tre^{*}(B) = (B
\otimes \Ie).
\end{align}
\end{subequations}
As such, for all $A \in \Mcal_{n} (\Cbb)$ and $B \in \Mcal_{n_{\rms}}
(\Cbb)$,
\begin{subequations}
\begin{align}
& \langle B \otimes \Ie, U A U^{\dag}\rangle =
\langle B \otimes \Ie, V A V^{\dag} \rangle, \\
\label{eqa:invariant}
& \Rightarrow \langle U^{\dag}(B \otimes \Ie) U, A \rangle =
\langle V^{\dag}(B \otimes \Ie)V, A \rangle.
\end{align}
Since (\ref{eqa:invariant}) is true for all $A \in \Mcal_{n}
(\Cbb)$ and $B \in \Mcal_{n_{\rms}} (\Cbb)$, then
\begin{align}
& U^{\dag}(B \otimes \Ie) U = V^{\dag} (B \otimes \Ie) V, \\
\label{eqa:general}
& \Rightarrow (B \otimes \Ie) \left( U V^{\dag} \right) = \left( U
 V^{\dag} \right) (B \otimes \Ie).
\end{align}
Finally, because (\ref{eqa:general}) is true for all $B \in
\Mcal_{n_{\rms}} (\Cbb)$, it holds for all $W \in \Us$:
\begin{equation}
\label{eqa:unitary}
(W \otimes \Ie) \left( UV^{\dag} \right) = \left( U V^{\dag} \right) (W
\otimes \Ie),
\end{equation}
\end{subequations}
i.e., $UV^{\dag}$ commutes with $W \otimes \Ie$, for all $W$. As such,
$UV^{\dag}$ is an element of the \emph{centralizer} of $\Us \otimes \Ie$
[i.e., the subgroup of $\Uc$ that commutes with $\Us \otimes \Ie$, also
formally defined in \ref{app:norm-cent}], denoted as
$\Ccal_{\Uc}(\Us \otimes \Ie)$ \cite{Dummit03a, Lang05a}. In
\ref{app:norm-cent}, we show that the centralizer of $\Us \otimes \Ie$
is $\Is \otimes \Ue$. Thus,
\begin{equation}
UV^{\dag} = \Is \otimes \Phi \in \Is \otimes \Ue.
\end{equation}
So, we conclude that $U = (\Is \otimes \Phi)V$ for some $\Phi \in \Ue$.
\end{proof}

\section{Chaining of unitary operations}
\label{app:chain}

After defining and establishing some mathematical objects and terms
necessary for our work, we show that the chaining property from
(\ref{eq:chain}) holds for a unitarily bi-invariant quotient metric on
the homogeneous space $\Qcal$.

\begin{definition}
Let $G$ be a group and $H$ be a subgroup of $G$. The \emph{normalizer}
of $H$ in $G$, $\Ncal_{G}[H] = \{ g \in G : g H g^{-1} = H \}$, is
the largest subgroup of $G$ such that $H$ is a normal subgroup of
$\Ncal_{G}[H]$.
\end{definition}

\begin{lemma}
The normalizer of $\Us \otimes \Ie$ in $\Uc$ is $\Us \otimes \Ue
:= \{ \Psi \otimes \Phi : \Psi \in \Us \ \textrm{and} \ \Phi \in
\Ue \}$. Likewise, the normalizer of $\Is \otimes \Ue$ in $\Uc$ is $\Us
\otimes \Ue$.
\end{lemma}

\begin{definition}
The \emph{centralizer} of $H$ in $G$, $\Ccal_{G}[H] = \{ g \in G :
g h g^{-1} = h, \forall h \in H\}$, is the subgroup of all elements in
$G$ that commute with all elements of $H$. Thus, $\Ccal_{G}[H]$ is
a subgroup of $\Ncal_{G}[H]$ \emph{\cite{Dummit03a, Lang05a}}.
\end{definition}

\begin{lemma}The centralizer of $\Us \otimes \Ie$ in $\Uc$ is $\Is
 \otimes \Ue := \{ \Ie \otimes \Phi : \Phi \in \Ue \}$. Likewise, the
 centralizer of $\Is \otimes \Ue$ in $\Uc$ is $\Us \otimes \Ie$.
\end{lemma}

Proofs of these lemmas are provided in \ref{app:norm-cent}. The theorem
from section~\ref{sec:quot-metric} is repeated here for convenience.

\begin{theorem}
If $\Delta$ is bi-invariant on $\Uc$, then
\begin{equation}
\label{eqa:chain}
\Delta \big( U_{m} \ldots U_{2} U_{1}, V_{m} \ldots V_{2} V_{1} \big)
\leq \Delta \left( U_{1}, V_{1} \right) + \Delta \big( U_{2}, V_{2}
\big) + \ldots + \Delta \big( U_{m}, V_{m} \big) 
\end{equation}
holds for any $U_{1}, U_{2}, \ldots, U_{m} \in \Uc$ and $V_{1}, V_{2},
\ldots, V_{m} \in \Us \otimes \Ue$.
\end{theorem}

\begin{proof}
We demonstrate that (\ref{eqa:chain}) holds for $m = 2$. This result
can be generalized to arbitrary values of $m$. If $\Delta$ is
bi-invariant on $\Uc$, i.e., for any $U, V, W_{1}, W_{2} \in \Uc$,
$\Delta(W_{1}UW_{2}, W_{1}VW_{2}) = \Delta(U, V)$, then, for any $U_{1},
U_{2}, V_{1}, V_{2} \in \Uc$,
\begin{subequations}
\label{eqa:chain-proof}
\begin{align}
\tilde{\Delta} \big( [U_{2}U_{1}], [V_{2}V_{1}] \big) & = \min_{\Phi}
\left\{ \Delta \big[ U_{2}U_{1}, (\Is \otimes \Phi) V_{2} V_{1} \big]
\right\}, \\
& \leq \min_{\Phi} \big\{ \Delta \big[ U_{2}U_{1}, U_{2} (\Is \otimes
 \Phi_{1})V_{1} \big] \nonumber \\
& \hspace*{0.48cm} + \Delta \big[ (U_{2}(\Is \otimes \Phi_{1})V_{1},
 (\Is \otimes \Phi) V_{2} V_{1} \big] \big\}, \\
& = \Delta \big[ U_{1}, (\Is \otimes \Phi_{1}) V_{1} \big] +
\min_{\Phi} \left\{ \Delta \big[ U_{2}, (\Is \otimes \Phi) V_{2} (\Is
 \otimes \Phi_{1}^{\dag}) \big] \right\}, \\ 
& \leq \Delta \big[ U_{1}, (\Is \otimes \Phi_{1}) V_{1} \big] +
\min_{\Phi} \left\{ \Delta \big[ U_{2}, (\Is \otimes \Phi \Phi_{2}) V_{2}
 \big] \right\}
\nonumber \\
& \hspace*{0.48cm} + \Delta \big[ (\Is \otimes \Phi_{2}) V_{2},
V_{2} (\Is \otimes \Phi_{1}^{\dag}) \big].
\end{align}
\end{subequations}
The final result of (\ref{eqa:chain-proof}) is obtained through repeated
usage of the bi-invariance of $\Delta$ and the triangle inequality.
Because this inequality holds for all $\Phi_{1}, \Phi_{2} \in \Ue$, it
holds for the $\Phi_{1}$ that minimizes $\Delta \big[ U_{1}, (\Is
\otimes \Phi_{1}) V_{1} \big]$. Now, if $V_{2} \in \Ncal_{\Uc}
[\Is \otimes \Ue]$, i.e., $V_{2}$ is an element of the normalizer of
$\Is \otimes \Ue$ in $\Uc$, then for each $\Phi_{1} \in \Ue$ there
exists a $\Phi_{2}\in \Ue$ such that $\Delta \big[ (\Is \otimes
\Phi_{2})V_{2},V_{2} (\Is \otimes \Phi_{1}^{\dag}) \big] = 0$. In
particular, such a $\Phi_{2}$ exists for the minimizing operator
$\Phi_{1}$. Therefore, under these conditions,
\begin{subequations}
\begin{align}
\tilde{\Delta} \big( [U_{2}U_{1}], [V_{2}V_{1}] \big) & \leq
\min_{\Phi_{1}} \left\{ \Delta \big[ U_{1}, (\Is \otimes \Phi_{1})V_{1}
 \big] \right\} + \min_{\Phi} \left\{ \Delta \big[ U_{2}, (\Is \otimes
 \Phi \Phi_{2})V_{2} \big] \right\}, \\
& = \tilde{\Delta} \big( [U_{1}], [V_{1}] \big) + \tilde{\Delta} \big(
[U_{2}], [V_{2}] \big).
\end{align}
\end{subequations}
\end{proof}

From this result, the chaining property holds when $V_{2} \in
\Ncal_{\Uc} [\Is \otimes \Ue] = \Us \otimes \Ue$, i.e., the set
of all decoupled or factorizable unitary operators in $\Uc$. This
condition is satisfied by any $V_{2}$ that preserves the isolation of
the system $\Scal$ from the environment $\Ecal$, a crucial
target in QC.

\section{\texorpdfstring{The normalizer and centralizer of 
$\bm{\Us \otimes \Ie}$ in $\bm{\Uc}$}
{The normalizer and centralizer}}
\label{app:norm-cent}

The lemmas from \ref{app:chain} are repeated here for convenience.

\setcounter{lemma}{0}
\begin{lemma}
The normalizer of $\Us \otimes \Ie$ in $\Uc$ is $\Us \otimes \Ue
:= \{ \Psi \otimes \Phi : \Psi \in \Us \ \textrm{and} \ \Phi \in
\Ue \}$. The normalizer of $\Is \otimes \Ue$ in $\Uc$ is also $\Us
\otimes \Ue$.
\end{lemma}

\begin{proof}
We determine the structure of the normalizer of $\Us \otimes \Ie$ in
$\Uc$; the structure of the normalizer of $\Is \otimes \Ue$ follows
from the symmetry of the tensor product. First, we define a fixed
non-degenerate diagonal unitary matrix $\Lambda \in \Us$. Note that if
$U \in \Ncal_{\Uc} [\Us \otimes \Ie]$, then in particular,
$U(\Lambda \otimes \Ie)U^{\dag} \in \Us \otimes \Ie$. For simplicity,
let $U(\Lambda \otimes \Ie)U^{\dag} = V \otimes \Ie$, where $V \in
\Us$. Because $U$ is unitary, $V \otimes \Ie$ and $\Lambda \otimes \Ie$
have the same eigenvalues (they are related by a similarity
transformation), hence $V$ and $\Lambda$ must have the same eigenvalues,
so there exists an $\Omega \in \Us$ such that $\Omega V \Omega^{\dag} =
\Lambda$. Therefore,
\begin{equation}
(\Omega \otimes \Ie) U (\Lambda \otimes \Ie) [(\Omega \otimes \Ie)
U]^{\dag} = \Lambda \otimes \Ie,
\end{equation}
demonstrating that $(\Omega \otimes \Ie)U$ is an element of the
stabilizer subgroup of $\Lambda \otimes \Ie$ with respect to the group
action of conjugation. Hence,
\begin{equation}
(\Omega \otimes \Ie) U \in \underbrace{\Ue \oplus \cdots \oplus
 \Ue}_{n_{\rms} \ \textrm{times}},
\end{equation}
i.e., $(\Omega \otimes \Ie)U$ is block diagonal with $n_{\rms}$
blocks, each of which are $n_{\rme}\times n_{\rme}$
unitary matrices. So we have shown that any $U$ in the normalizer can
be written as $U = (\Omega \otimes \Ie)V$ for some $\Omega \in \Us$
and $V \in \Ue \oplus \cdots \oplus \Ue$. We can write
\begin{equation}
V = \left( \begin{array}{cccc}
V_{1} \\ 
& V_{2} \\
& & \ddots \\
& & & V_{n_{\rms}}
\end{array} \right),
\end{equation}
where each $V_{i} \in \Ue$ for all $i$.

In order for $U$ to be in the normalizer subgroup, it must be true that
for each $W \in \Us$, there exists an $X \in \Us$ such that $U(W \otimes
\Ie) = (X \otimes \Ie)U$. If $U = (\Omega^{\dag} \otimes \Ie)V$, we must
have $V(W \otimes \Ie) = (\Omega X \Omega^{\dag} \otimes \Ie) V$. This
condition implies that $w_{ij}V_{i} = (\Omega X \Omega)_{ij}V_{j}$ for
all $1 \leq i, j \leq n_{\rms}$. Because $W$ is arbitrary, for each $i$
consider the case where $w_{i1} = 1$ (and $w_{j1} = 0$ for $j \neq i$),
implying that $V_{i}$ is a scalar times $V_{1}$. Hence, $V =
\mathrm{diag} \left[ \exp(\rmi \theta_{1}), \dots, \exp(\rmi
\theta_{n_{\rms}}) \right] \otimes \Phi\in \Us \otimes \Ue$. Therefore,
we conclude that $U = (\Omega^{\dag} \otimes \Ie) V \in \Us \otimes \Ue$
and $\Ncal_{\Uc} [\Us \otimes \Ie] \subset \Us \otimes \Ue$.

For any $U = \Psi \otimes \Phi \in \Us \otimes \Ue$, $U (\Us \otimes
\Ie)U^{\dag} = \Psi \Us \Psi^{\dag} \otimes \Ie = \Us \otimes \Ie$, so
that $\Ncal_{\Uc} [\Us \otimes \Ie] \supset \Us \otimes \Ue$.
Therefore, $\Ncal_{\Uc} [\Us \otimes \Ie] = \Us \otimes \Ue$
\end{proof}

\begin{lemma}The centralizer of $\Us \otimes \Ie$ in $\Uc$ is $\Ie
 \otimes \Ue$. Likewise, the centralizer of $\Ie \otimes \Ue$ in
 $\Uc$ is $\Us \otimes \Ie$.
\end{lemma}

\begin{proof}
We consider the centralizer of $\Us \otimes \Ie$ in $\Uc$, since the
centralizer of $\Ie \otimes \Ue$ follows from the symmetry of the tensor
product. Because the centralizer must be a subset of the normalizer
subgroup, each $U \in \Ccal_{\Uc} [\Us \otimes \Ie]$ must be of
the form $U = \Psi \otimes \Phi$ where $\Psi \in \Us$ and $\Phi \in
\Ue$. So,
\begin{align}
\Ccal_{\Uc} [\Us \otimes \Ie] = & \
 \{ \Psi \otimes \Phi : \Psi \in \Us, \Phi \in \Ue
\nonumber \\ & \ \ \textrm{and} \
(\Psi \otimes \Phi)(V \otimes \Ie)(\Psi \otimes \Phi)^{\dag} = V \otimes
\Ie, \ \forall V \in \Us \}.
\end{align}
Now, $(\Psi \otimes \Phi) (V \otimes \Ie) (\Psi \otimes \Phi)^{\dag} =
(\Psi V \Psi^{\dag}) \otimes \Ie$, so the only condition on $\Psi$ and
$\Phi$ is that $\Psi$ must lie in the center of $\Us$, i.e., $\Psi =
\exp(\rmi \theta)\Ie$ and $\Phi \in \Ue$. Then, $\Psi \otimes
\Phi = \exp(\rmi\theta) \Ie \otimes \Phi = \Ie \otimes \left[
\exp(\rmi\theta) \Phi \right] \in \Ie \otimes \Ue$. Therefore,
$\Ccal_{\Uc} [\Us \otimes \Ie] \subset \Ie \otimes \Ue$. Because
every element of $\Ie \otimes \Ue$ clearly commutes with $\Us \otimes
\Ie$, we conclude that $\Ccal_{\Uc} [\Us \otimes \Ie] = \Ie
\otimes \Ue$. 
\end{proof}

\section{Centrality of the maximally-mixed state}
\label{app:rho-min-max}

We show that $\rho_{\rmm, \rms} := \Is/n_{\rms}$ is the density matrix
with the ``shortest'' maximum distance, under the Hilbert-Schmidt norm,
to all other density matrices, i.e.,
\begin{subequations}
\begin{gather}
\rho_{\rmm, \rms} = \arg \min_{\rho_{2}} \left( \max_{\rho_{1}} \|
\rho_{1} - \rho_{2} \|_{\rHS}^{2} \right), \ \ \forall \rho_{1},
\rho_{2} \in \Dcal(\Hs), \\ 
\intertext{where}
\min_{\rho_{2}} \left( \max_{\rho_{1}} \| \rho_{1} - \rho_{2}
\|_{\rHS}^{2} \right) = \frac{n_{\rms}-1}{n_{\rms}}.
\end{gather}
\end{subequations}

\begin{proof}
Let $\rho_{2} = \sum_{i} \zeta_{i} |i\rangle\langle i|$ with $0 \leq
\zeta_{1} \leq \dots \leq \zeta_{n}\leq 1$, and define the
Hilbert-Schmidt distance as
\begin{equation}
f(\rho_{1}) := \| \rho_{2} - \rho_{1} \|_{\rHS}^{2} = \| \rho_{2}
\|_{\rHS}^{2} + \| \rho_{1} \|_{\rHS}^{2} - 2\Tr(\rho_{2} \rho_{1}). 
\end{equation}
If $\rho_{1} = |1\rangle\langle 1|$, then $f(\rho_{1}) = \| \rho_{2}
\|_{\rHS}^{2} + 1 - 2\zeta_{1}$. Thus, $\max_{\rho_{1}} f(\rho_{1}) \geq
\|\rho_{2}\|_{\rHS}^{2} + 1 - 2\zeta_{1}$. Now, $\Tr(\rho_{2}\rho_{1}) =
\sum_{i} \zeta_{i}\langle i|\rho_{1}|i\rangle$ for all $\rho_{1}$, where
$\sum_{i} \langle i| \rho_{1} |i\rangle = \Tr (\rho_{1}) = 1$, so
$\Tr(\rho_{2}\rho_{1}) \geq \zeta_{1}$. Also, $\| \rho_{1} \|_{\rHS}^{2}
\leq 1$ in general, which implies that $f(\rho_{1}) \leq \| \rho_{2}
\|_{\rHS}^{2} + 1 - 2\zeta_{1}$ for all $\rho_{1}$. Therefore,
\begin{equation}
\max_{\rho_{1}} \| \rho_{2} - \rho_{1} \|_{\rHS}^{2} = \| \rho_{2}
\|_{\rHS}^{2} + 1 - 2 \zeta_{1}.
\end{equation}
To minimize this expression over $\rho_{2}$, we want to minimize $\|
\rho_{2} \|_{\rHS}^{2}$ and maximize $\zeta_{1}$. Because $\zeta_{1}$
is the minimum eigenvalue of $\rho_{2}$, its maximum value of
$1/n_{\rms}$ occurs when $\rho_{2} = \Is/n_{\rms}$. Likewise, $\|
\rho_{2} \|_{\rHS}^{2}$ takes its minimum value $1/n_{\rms}$ at
$\rho_{2} =\Is/n_{\rms}$. So, indeed the argument of $\min_{\rho_{2}}
\max_{\rho_{1}} \| \rho_{2} - \rho_{1} \|_{\rHS}^{2}$ is $\rho_{\rmm,
 \rms} = \Is/n_{\rms}$ and its value is $(n_{\rms} - 1)/n_{\rms}$.
\end{proof}

\section{Relating the distance measure to fidelity}
\label{app:fidelity}

We demonstrate that the inequality relation between the distance and
fidelity of quantum states presented in (\ref{eq:dist-fidelity1})
\cite{Fuchs99a} also applies to the distance and fidelity of quantum
operations given by $\tilde{\Delta}_{\rHS} ([U], [\Vsys \otimes \Ie])$
in (\ref{eq:HS-dist-c}) and $\Fcal_{\mathrm{m,c}}(\Kcal,
\Vsys)$ in (\ref{eq:fidelity-mixed}):
\begin{equation}
\label{eqa:dist-fidelity}
\left[ 1 - \tilde{\Delta}_{\rHS} ([U], [\Vsys \otimes \Ie]) \right]^{2}
\leq \Fcal_{\mathrm{m,c}}(\Kcal, \Vsys) \leq 1 -
\tilde{\Delta}_{\rHS} ([U], [\Vsys \otimes \Ie])^{2}.
\end{equation}

\begin{proof}
We first show that the bounds on
$\Fcal_{\mathrm{m,c}}(\Kcal, \Vsys)$ follow from $\| \Gamma
\|_{2} = \max_{i} \sigma_{i} \leq n_{\rms}$, where $\sigma_{i}$ is the
$i$th singular value of $\Gamma$. Then we show that $\| \Gamma \|_{2}
\leq n_{\rms}$ holds.

The upper bound holds if and only if $\| \Gamma \|_{\rHS}^{2}
\leq n_{\rms} \| \Gamma \|_{\Tr}$, or equivalently, $\sum_{i}
\sigma_{i}^{2} \leq n_{\rms} \sum_{i} \sigma_{i}$. Assuming that
$\max_{i} \sigma_{i} \leq n_{\rms}$, it immediately follows that
$\sum_{i} \sigma_{i}^{2} \leq n_{\rms} \sum_{i} \sigma_{i}$, which
establishes the upper bound.

With some algebraic rearrangement, the lower bound becomes $\| \Gamma
\|_{\Tr} + (1/n_{\rms})\| \Gamma \|_{\rHS}^{2} \leq 2\sqrt{n_{\rme}}\|
\Gamma \|_{\rHS}$. If $\| \Gamma \|_{2} \leq n_{\rms}$, which implies
that $\| \Gamma \|_{\rHS}^{2} \leq n_{\rms} \| \Gamma \|_{\Tr}$ and $\|
\Gamma \|_{\Tr} + 1/n_{\rms} \| \Gamma \|_{\rHS}^{2} \leq 2 \| \Gamma
\|_{\Tr}$, the lower bound holds if $\| \Gamma \|_{\Tr} \leq
\sqrt{n_{\rme}} \| \Gamma \|_{\rHS}$. This is a known inequality for any
matrix and thus establishing the lower bound. To prove this inequality,
let $\vec{\sigma}$ and $\vec{\tau}$ be two $k$-dimensional vectors
(where $k \leq n_{\rme}$) with entries $\sigma_{i}$, which are the
singular values of $\Gamma$, and $\tau_{i} = 1$ for all $i$,
respectively. It follows from the Cauchy-Schwarz inequality that $\|
\Gamma \|_{\Tr} = \sum_{i} \sigma_{i} = \| \vec{\sigma} \|_{1} = \langle
\vec{\sigma}, \vec{\tau} \rangle \leq \| \vec{\sigma} \|_{2} \|
\vec{\tau} \|_{2} = \sqrt{k} \| \vec{\sigma} \|_{2} = \sqrt{k} \| \Gamma
\|_\rHS \leq \sqrt{n_{\rme}} \| \Gamma \|_\rHS$, where, for a vector,
$\| \cdot \|_{2}$ is the Euclidean norm (also referred to as the
$\ell_{2}$ norm):
\begin{equation}
\| \vec{\sigma} \|_{2} := \left( \sum_{i = 1}^{k} \sigma_{i}^{2}
\right)^{2}.
\end{equation}

To show that $\sigma_{i} \leq n_{\rms}$, consider
\begin{subequations}
\begin{align}
\label{eqa:2-norm1}
\| \Gamma |\psi\rangle \|_{2} & = \| \left\{ \Trs \left[ U \left(
\Vsys^{\dag} \otimes \Ie \right) \right] \right\} |\psi \rangle
\|_{2}, \\
\label{eqa:2-norm2}
& = \| \Trs \left[ U ( \Vsys^{\dag} \otimes \Ie ) (\Is \otimes |\psi
 \rangle) \right] \|_{2},
\end{align}
where $\dim \{ |\psi\rangle \} = n_{\rme}$ and $\langle\psi |\psi\rangle
= 1$. Applying the Cauchy-Schwarz inequality \cite{Horn90a} to the
argument of (\ref{eqa:2-norm2}) yields
\begin{equation}
\| \Gamma |\psi \rangle \|_{2} \leq \sqrt{n_{\rms}} \| U \left(
 \Vsys^{\dag} \otimes \Ie \right) (\Is \otimes |\psi \rangle) \|_{2}.
\end{equation}
Because the two-norm is bounded from above by the Hilbert-Schmidt
norm \cite{Horn90a},
\begin{align}
\| \Gamma |\psi\rangle \|_{2} & \leq \sqrt{n_{\rms}} \| U \left(
 \Vsys^{\dag} \otimes \Ie \right) (\Is \otimes |\psi
 \rangle) \|_{\rHS}, \\
& = \sqrt{n_{\rms}} \| \Is \otimes |\psi\rangle \|_{\rHS}, \\
& = n_{\rms} \langle\psi |\psi \rangle = n_{\rms}.
\end{align}
\end{subequations}
Since $\| \Gamma |\psi\rangle \|_{2} \leq n_{\rms}$ is true for all
normalized $|\psi\rangle$, it follows that $\| \Gamma \|_{2} \leq
n_{\rms}$.
\end{proof}

\section{\texorpdfstring{Optimal control objective functional: 
$\bm{\J}$}{Optimal control objective functional}}
\label{app:objective}

In section~\ref{ssec:algorithm} we use the objective functional $\Jsys$
in (\ref{eq:functional}), which corresponds to a special situation
(although ubiquitous in QC), where the target unitary transformation
$\Vsys \in \Us$ is specified for a closed system (i.e., for the system
$\Scal$). As discussed in section~\ref{ssec:closedsystem-target},
in this situation the target unitary operator $V \in \Uc$ should be
replaced by the tensor product $\Vsys \otimes \Ie$ in the quotient
metric, leading to the Hilbert-Schmidt norm distance
$\tilde{\Delta}_{\rHS} ([U], [\Vsys \otimes \Ie])$ of
(\ref{eq:HS-dist-c}). However, in this appendix, we consider a more
general case: an arbitrary target unitary operator $V \in \Uc$, not
necessarily the tensor-product one. The corresponding objective
functional $\J$ is
\begin{equation}
\label{eq:functional-general}
\J(C) := \tilde{\Delta}_{\rHS} \big( [U_{\tf} (C)], [V] \big)
+ \frac{\alpha}{2} \| C \|_{\Kbb_{\tf}}^{2},
\end{equation}
where the Hilbert-Schmidt norm distance $\tilde{\Delta}_{\rHS} ([U],
[V])$ is defined by (\ref{eq:HS-dist-b}) and (\ref{eq:gamma}).

\subsection{\texorpdfstring{Computing $\grad \tilde{\Delta} ([U],
 [V])_{\rHS}$}{Computing the gradient of the distance}}

Let $\Hm_{n}$ denote the space of $n \times n$ Hermitian matrices, i.e.,
$\Hm_{n} = \{A \in \Mcal_{n}(\Cbb) \ : \ A^{\dag} = A \}$. We endow the
linear spaces $\Hm_{n}$ and $\Mcal_{n}(\Cbb)$ with the real
Hilbert-Schmidt inner product in (\ref{eq:HS-in-prod}), and $\Uc$ is
given the Riemannian metric induced by this inner product. Define maps
$y$, $Z$, and $\Gamma_{V}$ by
\begin{equation}
\begin{array}{lcl}
y : \Hm_{n_{\rme}} \to \Rbb & \qquad & y(Z) := \sqrt{1 - n^{-1} \Tr(Z)},
\\
Z : \Mcal_{n{_\rme}} (\Cbb) \to \Hm_{n_{\rme}} & & Z(\Gamma) :=
\sqrt{\Gamma^{\dag} \Gamma}, \\
\Gamma_{V} : \Uc \to \Mcal_{n_{\rme}} (\Cbb) & & \Gamma_{V}(U) := \Trs
\left[ UV^{\dag} \right].
\end{array}
\end{equation}
For any fixed equivalence class $[V] \in \Qcal$, let
$\Delta_{\rHS}^{[V]}: \Uc \to \Rbb$ be defined by
$\Delta_{\rHS}^{[V]}(U) :=\tilde{\Delta}_{\rHS}([U], [V])$. Then
$\Delta_{\rHS}^{[V]}(U) = y \circ Z \circ \Gamma_{V}(U)$, and the
differential of $\Delta_{\rHS}^{[V]}$ at the point $U \in \Uc$ in the
direction $\delta U \in T_{U} \Uc$ is given by the chain rule:
\begin{equation}
\label{eqna:diff-chain}
\rmd_{U} \Delta_{\rHS}^{[V]} (\delta U) = \rmd_{Z \left[ \Gamma_{V} (U)
 \right]} y \circ \rmd_{\Gamma_{V} (U)} Z \circ \rmd_{U}
\Gamma_{V}(\delta U).
\end{equation}
The differential and the gradient are related through the Riemannian
metric by $\rmd_{U} \Delta_{\rHS}^{[V]} (\delta U) = \langle \grad
\Delta_{\rHS}^{[V]} (U), \delta U \rangle$ for all $\delta U \in
T_{U}\Uc$ \cite{doCarmo92a}. Using (\ref{eqna:diff-chain}), we
can also write
\begin{subequations}
\begin{align}
\rmd_{U} \Delta_{\rHS}^{[V]} (\delta U) & = \Big\langle \grad
y \big\{ Z \left[ \Gamma_{V}(U) \right] \big\}, \ \rmd_{\Gamma_{V}(U)} Z
\circ \rmd_{U}\Gamma_{V}(\delta U) \Big\rangle, \\
& = \Big\langle \big(\rmd_{U} \Gamma_{V}\big)^{*} \circ \big(
\rmd_{\Gamma_{V}(U)} Z \big)^{*} \Big( \grad y \big\{ Z \left[
 \Gamma_{V}(U) \right] \big\} \Big), \ \delta U \Big\rangle,
\end{align}
\end{subequations}
where $\grad y \big\{ Z \left[ \Gamma_{V}(U) \right] \big\}$ denotes the
gradient of $y$ at the point $Z \left[ \Gamma_{V}(U) \right]$ and $^{*}$
denotes the operator adjoint. So, we can compute the desired gradient as
\begin{equation}
\label{eqna:grad-chainrule}
\grad \Delta_{\rHS}^{[V]} (U) = \big( \rmd_{U} \Gamma_{V} \big)^{*}
\Big\{ \left( \rmd_{\Gamma_{V}(U)}Z \right)^{*} \big[ \grad y \big\{ Z
\left[ \Gamma_{V}(U) \right] \big\} \big] \Big\}.
\end{equation}
To do this, we must find expressions for the gradient and two adjoint
operators on the right hand side of (\ref{eqna:grad-chainrule}). We
develop each of these differentials separately.

We begin by first computing $\grad y$. For $y(Z) = \sqrt{1 - n^{-1}
\Tr(Z)}$, the argument of the square root is an affine scalar
function of $Z$, so $\rmd_{Z} \, y(\delta Z) = -\Tr(\delta Z)/[2n \,
y(Z)]$. The gradient is then given by $\rmd_{Z} \, y(\delta Z) = \langle
\grad y(Z), \delta Z \rangle = \Real \left\{ \Tr \big[ \grad y(Z) \delta
 Z \big] \right\} = \Tr \big[ \grad y(Z) \delta Z \big]$. Hence,
\begin{equation}
\grad \, [y(Z)] = -\frac{\openone}{[2n \, y(Z)]}.
\end{equation}

We now turn to the problem of differentiating $Z$. Differentiating a
square root in a commutative setting is straight-forward; we used this
fact to differentiate $y$. But for matrices, non-commutivity may cause
trouble. So, we begin by rewriting the definition of $Z$ as $Z^{2} =
\Gamma^{\dag} \Gamma$, and differentiating both sides:
\begin{equation}
\label{eqna:diffZ}
Z \rmd_{\Gamma} Z(\delta \Gamma) + \rmd_{\Gamma} Z(\delta \Gamma) Z =
\delta \Gamma^{\dag} \Gamma + \Gamma^{\dag} \delta \Gamma.
\end{equation}
Define $\nu : \Mcal_{n_{\rme}} (\Cbb) \to \Cbb^{n_{\rme}^{2}}$ to be the
linear operator that stacks the columns of a square matrix to create a
vector. If we let $\Mcal_{n_{\rme}} (\Cbb)$ and $\Cbb^{n_{\rme}^{2}}$
both be \emph{real} Hilbert spaces with respective inner products
$\langle A, B\rangle = \Real \left[ \Tr \left( A^{\dag}B \right)
\right]$ and $\langle \mathbf{x},\mathbf{y}\rangle = \Real \left(
\mathbf{x}^{\dag} \mathbf{y} \right)$, then $\nu$ is a linear isometry
\cite{Naylor82a} between these two spaces:
\begin{equation}
\langle A, B\rangle = \sum_{i,j}\Real \left( A_{ij} \right) \Real
\left( B_{ij} \right) + \Imag \left( A_{ij} \right) \Imag \left( B_{ij}
\right) = \langle \nu(A), \nu(B) \rangle,
\end{equation}
for all $A, B \in \Mcal_{n_{\rme}} (\Cbb)$. Hence $\nu^{*} = \nu^{-1}$
is the ``matrix-ization'' operator that returns a square matrix.

From the discussion of matrix equations and the Kronecker product in
\cite{Horn91a}, we can rewrite (\ref{eqna:diffZ}) as 
\begin{equation}
\left[ Z^{\rmT} \otimes \openone + \openone \otimes Z \right] \nu \left[
\rmd_{\Gamma} Z (\delta \Gamma) \right] = \nu \big( \delta
\Gamma^{\dag} \Gamma + \Gamma^{\dag} \delta \Gamma\big).
\end{equation}
By construction, $Z$ will be Hermitian and positive semi-definite; the
transpose $Z^{\rmT}$ will also be Hermitian and will have the same
eigenvalue decomposition as $Z$. Because of this, we know that $Z^{\rmT}
\otimes \openone + \openone \otimes Z$ will also be Hermitian and
positive semi-definite (the eigenvalues of the Kronecker sum are just
the set of all pairwise sums of the eigenvalues of $Z$). If we assume
that $Z$ is positive \emph{definite}, then $Z^{\rmT} \otimes \openone +
\openone \otimes Z$ will also be positive definite and this
linear system will have a unique solution:
\begin{equation}
\rmd_{\Gamma} Z (\delta \Gamma) = \nu^{*} \Big[ \big( Z^{\rmT} \otimes
\openone + \openone \otimes Z \big)^{-1} \nu \big( \delta \Gamma^{\dag}
\Gamma + \Gamma^{\dag} \delta \Gamma \big) \Big]
\end{equation}
Let $\Lcal_{\Gamma} : \Mcal_{n_{\rme}} (\Cbb) \to \Mcal_{n_{\rme}}
(\Cbb)$ be given by $\Lcal_{\Gamma} (A) = A^{\dag} \Gamma +
\Gamma^{\dag} A$. Then, considering $\Mcal_{n_{\rme}} (\Cbb)$ to be a
\emph{real} Hilbert space with the real Hilbert-Schmidt inner product,
$\Lcal_{\Gamma}$ is a linear operator, and we can find its adjoint:
\begin{subequations}
\begin{align}
\left\langle A, \Lcal_{\Gamma}^{*} (B) \right\rangle & =
\left\langle \Lcal_{\Gamma}(A), B \right\rangle = \left\langle
A^{\dag}\Gamma + \Gamma^{\dag}A, B \right\rangle, \\
& = \Real \left[ \Tr \left( \Gamma^{\dag} AB + A^{\dag} \Gamma B \right)
 \right], \\ 
& = \Real \left\{ \Tr \left[ \left( \Gamma B^{\dag} \right)^{\dag} A +
 (\Gamma B)^{\dag} A \right] \right\} = \left\langle A, \Gamma \left(
 B + B^{\dag} \right) \right\rangle,
\end{align}
\end{subequations}
for a given $B \in \Mcal_{n_{\rme}} (\Cbb)$ and all $A \in
\Mcal_{n_{\rme}} (\Cbb)$. Hence, $\Lcal_{\Gamma}^{*}(B) = \Gamma \left(
  B + B^{\dag} \right)$. Then,
\begin{subequations}
\begin{align}
\rmd_{\Gamma} Z & = \nu^{*} \circ \big( Z^{\rmT} \otimes \openone +
\openone \otimes Z \big)^{-1} \circ \nu \circ \Lcal_{\Gamma}, \\
\intertext{and}
\rmd_{\Gamma} Z^{*} & = \Lcal_{\Gamma}^{*} \circ \nu^{*} \circ
\big( Z^{\rmT} \otimes \openone + \openone \otimes Z \big)^{-1} \circ \nu.
\end{align}
\end{subequations}
Using $\grad y(Z) = -\openone/[2n \, y(Z)]$ derived previously yields
\begin{equation}
\rmd_{\Gamma} Z^{*} \big[ \grad y(Z) \big] = -\frac{1}{2n \,
 y(Z)}\Lcal_{\Gamma}^{*} \circ \nu^{*} \circ \big( Z\otimes
\openone + \openone \otimes Z \big)^{-1} \circ \nu(\openone).
\end{equation}
Note that $\nu^{*} \circ \left( Z^{\rmT} \otimes \openone + \openone
\otimes Z \right)^{-1} \circ \nu(\openone)$ is just the matrix
solution to the problem $ZA + AZ = \openone$, and therefore is $A =
\frac{1}{2}Z^{-1}$, since we are assuming $Z$ is not just positive
semi-definite, but positive \emph{definite}. Then,
\begin{equation}
\rmd_{\Gamma}Z^{*} \big[ \grad y(Z) \big] = -\frac{1}{4n \,
 y(Z)} \Lcal_{\Gamma}^{*} \left( Z^{-1} \right) = -\frac{1}{2n \,
 y(Z)} \Gamma Z^{-1}.
\end{equation}

For the differential of $\Gamma_{V}$, we can extend $\Gamma_{V}$ to
$\tilde{\Gamma}_{V} : \Mcal_{n} (\Cbb) \to \Mcal_{n_{\rme}} (\Cbb)$, by
letting $\tilde{\Gamma}_{V}(A) = \Trs \left( A \, V^{\dag}
\right)$. Then, $\tilde{\Gamma}_{V}$ is a linear map, so the
differential of $\tilde{\Gamma}_{V}$ is just itself, i.e., $\rmd_{A}
\tilde{\Gamma}(\delta A) = \tilde{\Gamma}_{V}(\delta A)$. For any
$\delta U \in T_{U} \Uc$, $\rmd_{U} \Gamma_{V}(\delta U) = \rmd_{U}
\tilde{\Gamma}_{V}(\delta U) = \tilde{\Gamma}_{V}(\delta U)$. We can
find the adjoint of $\rmd_{U} \Gamma_{V}$ by
\begin{subequations}
\begin{align}
\left\langle \left[ \rmd_{U} \Gamma_{V} \right]^{*} (B), \delta U
\right\rangle & =
\left\langle B, \rmd_{U} \Gamma_{V} (\delta U) \right\rangle =
\Real \Big\{ \Tre \big[ B^{\dag} \rmd_{U} \Gamma_{V}(\delta U) \big]
\Big\}, \\
& = \Real \Big\{ \Tre \big[ B^{\dag} \Trs \left( \delta U \, V^{\dag}
\right) \big] \Big\}, \\ 
& = \Real \Big\{ \Tr \big[ \delta U \, V^{\dag} \left( \openone \otimes
B^{\dag} \right) \big] \Big\}, \\
& = \frac{1}{2} \left\langle \big[ (\openone \otimes B) V - U
V^{\dag} \left( \openone \otimes B^{\dag} \right) U \big], \delta U
\right\rangle,
\end{align}
\end{subequations}
where the last step involves projecting $(\openone \otimes B) V$ onto
$T_{U} \Uc$. Hence, $\left( \rmd_{U} \Gamma_{V} \right)^{*}(B)
= \frac{1}{2} \big[ (\openone \otimes B) V - U V^{\dag} \left( \openone
\otimes B^{\dag} \right) U \big]$ for any $B \in \Mcal_{n_{\rme}}
(\Cbb)$.

Combining these components, we find that
\begin{subequations}
\begin{align}
\label{eqna:gradJ}
\grad \Delta_{\rHS}^{[V]} (U) & = \left( \rmd_{U} \Gamma_{V}
\right)^{*} \Big[ \left( \rmd_{\Gamma_{V}(U)} Z \right)^{*} \Big( \grad
y \big\{ Z \left[ \Gamma_{V}(U) \right] \big\} \Big) \Big], \\ 
& = -\frac{1}{4n \, \Delta_{\rHS}^{[V]}(U)} \left( R - UR^{\dag} U
\right),
\end{align}
\end{subequations}
where $R = \left( \openone \otimes \left\{ \Gamma_{V}(U)
Z^{-1} \left[ \Gamma_{V}(U) \right] \right\} \right) V = \big[ \openone
\otimes \left( \Omega W^{\dag} \right) \big] V$ and $\Gamma_{V}(U) =
\Omega SW^{\dag}$ is the singular value decomposition of
$\Gamma_{V}(U)$. Since $V \in \Uc$ and $\Omega W^{\dag} \in \Ue$, we
have $R \in \Uc$. Note that $R$ is the element of the equivalence class
$[V]$ that is closest to $U$ under the Hilbert-Schmidt distance, i.e.,
$R = (\Is \otimes \hat{\Phi})V$, where $\hat{\Phi} := \arg \min_{\Phi}
\big\{ \|U - (\Is \otimes \Phi) V \|_{\rHS} \big\}$.

\subsection{Explicitly coupling the Schr\"{o}dinger equation to the
 objective functional}

In this section, we consider the problem of coupling the Schr\"{o}dinger
equation to the kinematic cost function $\Delta_{\rHS}^{[V]}$ to arrive
at a gradient flow through the space of controls $\Kbb_{\tf}$
leading to an optimal control. To do this, we need to clarify the
definition and geometry of $\Kbb_{\tf}$. Let $\eta : [0, \tf]
\to \Rbb$, the ``shape function'', be positive almost everywhere and
continuous (and therefore bounded). Define the inner product on
$\Kbb_{\tf}$ as
\begin{equation}
\label{eqa:control-space}
\langle f, g \rangle_{\Kbb_{\tf}} := \int_{0}^{\tf} \frac{f(t)
 g(t)}{\eta(t)} \, \rmd t,
\end{equation}
where $\Kbb_{\tf}$ is the set of (equivalence classes of)
functions $f \in L^{2}([0,\tf]; \Rbb)$ such that $\| f \|_{\Kbb_{\tf}} <
\infty$. Bilinearity, symmetry, and non-negativity of the inner product
$\langle \cdot, \cdot \rangle_{\Kbb_{\tf}}$ are clear from the
definition. When $\eta(t) \equiv 1$, $\Kbb_{\tf}$ is $L^{2}([0,\tf];
\Rbb)$ with the standard inner product. Other choices of $\eta(t)$
change the standard geometry, moving undesirably-shaped functions far
away from the origin, or out to infinity, where they are less likely to
be the targets of an optimization on $\Kbb_{\tf}$. This point is
discussed in more detail throughout this section. For the optimizations
performed in this work, $\eta(t) = \sin(\pi t/\tf)$.

To make this approach (and geometric interpretation) rigorous, we need
to verify that $\Kbb_{\tf}$ is indeed a Hilbert space under the
inner product $\langle \cdot, \cdot \rangle_{\Kbb_{\tf}}$ in
(\ref{eqa:control-space}). Suppose $f$ is a measurable function such
that $\|f\|_{\Kbb_{\tf}} < \infty$. Because $\eta$ is positive
almost everywhere and bounded above, $1/\eta$ is positive everywhere and
bounded below by some $\beta > 0$. Then,
\begin{equation}
\|f\|_{\Kbb_{\tf}}^{2} = \int_{0}^{\tf} \frac{f^{2}(t)}{\eta(t)} \,
\rmd t \ > \ \beta \|f\|_{L^{2}}^{2} \ := \ \beta \int_{0}^{\tf}
f^{2}(t)\, \rmd t.
\end{equation}
Hence, $\|f\|_{L^{2}} < \infty$, so that $f \in L^{2}([0,\tf];
\Rbb)$, where $\| \cdot \|_{L^{2}}$ denotes the $L^{2}$ norm on $[0,
\tf]$. If $f(t) \neq 0$ on a set of non-zero measure, then $f^{2}(t) >
0$ on a set of non-zero measure [analogously for $\eta(t)$ and $f^{2}(t)
/ \eta(t)$]. Hence, $\|f\|_{\Kbb_{\tf}} > 0$, so we have non-degeneracy,
and $\Kbb_{\tf}$ is thus an inner product space.

Suppose that $\{f_{k}\}$ is a Cauchy sequence in $\Kbb_{\tf}$
\cite{Naylor82a}. Then $g_{k} := f_{k}/\sqrt{\eta}$ is a Cauchy sequence
in $L^{2}$, i.e., $\| g_{k} - g_{\ell} \|_{L^{2}} = \| f_{k} - f_{\ell}
\|_{\Kbb_{\tf}} \to 0$. Since $L^{2}([0,\tf]; \Rbb)$ is complete, this
implies that $g_{k} \to g \in L^{2}([0,\tf]; \Rbb)$. Let $g := f /
\sqrt{\eta}$, thus,
\begin{subequations}
\begin{align}
\|f\|_{\Kbb_{\tf}}^{2} & = \int_{0}^{\tf} f^{2}(t) / \eta(t) \, \rmd
t = \int_{0}^{\tf} g^{2}(t)\, \rmd t = \|g\|_{L^{2}}^{2} < \infty, \\
\intertext{and}
\| f_{k} - f \|_{\Kbb_{\tf}}^{2} & = \int_{0}^{\tf} \frac{ \big[
f_{k}(t) - f(t) \big]^{2}}{\eta(t)} \, \rmd t, \\
& = \int_{0}^{\tf} \big[ g_{k}(t) - g(t) \big]^{2} \, \rmd t = \| g_{k}
- g \|_{L^{2}}^{2} \to 0,
\end{align}
\end{subequations}
so that $f \in \Kbb_{\tf}$ and $f_{k} \to f$ in the
topology induced by $\langle \cdot, \cdot \rangle_{\Kbb_{\tf}}$.
Hence, $\Kbb_{\tf}$ with $\langle \cdot, \cdot
\rangle_{\Kbb_{\tf}}$ is a complete inner product space, i.e., a
Hilbert space. We have not simply effected a change in geometry, but we
also have restricted the domain to a linear subspace of $L^{2}([0,\tf];
\Rbb)$. If $C \in L^{2}([0,\tf]; \Rbb) / \Kbb_{\tf}$, then $\J =
\infty$, so that $C$ is not a viable solution. This new geometry puts
desirable functions near the origin and undesirable functions out near
infinity, which makes this optimization different from the unshaped
version.

Now, consider the mapping $U_{\tf} : \Kbb_{\tf} \to \Uc$
introduced in section~\ref{ssec:algorithm}. The differential of the
final-time evolution operator $U_{\tf}$ with respect to the control
field depends on evolution operators at all times in $[0,\tf]$. It is
shown in \cite{Dominy08a} that
\begin{equation}
\rmd_{C} U_{\tf}(\delta C) = \rmi U_{\tf}(C) \int_{0}^{\tf} U^{\dag}(t;
C) \mu U(t; C) \delta C(t) \, \rmd t,
\end{equation}
where $\mu$ is the dipole moment operator [e.g., $\mu =
\sum_{i}\mu_{i}S_{ix}$ in (\ref{eq:Hamiltonian2})]. With the usual
Hilbert-Schmidt inner product as the Riemannian metric on $\Uc$, we can
find the adjoint of this differential: for any $A \in T_{U_{\tf}(C)}
\Uc$ and all $\delta C \in T_{C} \Kbb_{\tf} \simeq
\Kbb_{\tf}$,
\begin{subequations}
\label{eqna:dV-adj-deriv}
\begin{align}
\left\langle \big( \rmd_{C}U_{\tf} \big)^{*}(A), \delta C
\right\rangle_{\Kbb_{\tf}} & = \left\langle A, \rmd_{C}
 U_{\tf}(\delta C) \right\rangle_{T\Uc}, \\
& = \Real \left\{ \Tr \left[ \rmi A^{\dag} U_{\tf}(C) \int_{0}^{\tf}
 U^{\dag}(t; C) \mu U(t; C) \delta C(t) \, \rmd t \right] \right\}, \\
& = \int_{0}^{\tf} \Real \bigg\{ \Tr \Big[ \rmi A^{\dag} U_{\tf}(C)
 U^{\dag}(t; C) \mu U(t; C) \Big] \delta C(t) \, \rmd t \bigg\}, \\
& = \left\langle \eta \, \Real \Big\{ \Tr \big[ \rmi A^{\dag} U_{\tf}(C)
 U^{\dag}(\cdot; C) \mu U(\cdot; C) \big] \Big\}, \delta C
\right\rangle_{\Kbb_{\tf}}.
\end{align}
\end{subequations}
So that 
\begin{equation}
\left[ \left( \rmd_{C} U_{\tf} \right)^{*} (A) \right](t) = -\eta(t)
\Imag \big\{ \Tr \left[ A^{\dag} U_{\tf}(C) U^{\dag}(t; C) \mu U(t; C)
\right] \big\}.
\end{equation}

With $\J(C)$ in (\ref{eq:functional-general}), using previous
arguments regarding the chain rule on gradients, we obtain
\begin{subequations}
\label{eqa:gradient}
\begin{align}
& \big( \grad \J(C) \big)(t) = \Big[ \big( \rmd_{C} U_{\tf} \big)^{*}
\big( \grad \Delta_{\rHS}^{[V]} [U_{\tf}(C)] \big) \Big](t) + \alpha
C(t), \\
\nonumber 
& = \frac{\eta(t)}{4n \, \Delta_{\rHS}^{[V]} [U_{\tf}(C)]} \Imag \Big(
\Tr \Big\{ \big[ R^{\dag} - U^{\dag}(\tf; C) R U^{\dag}(\tf; C) \big]
U_{\tf}(C) U^{\dag}(t; C) \mu U(t; C) \Big\} \Big) \\ 
& \hspace{0.48cm} + \alpha C(t), \\
& = \frac{\eta(t)}{4n \, \Delta_{\rHS}^{[V]} [U_{\tf}(C)]} \Imag \Big(
\Tr \Big\{ \big[ R^{\dag} U_{\tf}(C) - U^{\dag}(\tf; C) R \big]
U^{\dag}(t; C) \mu U(t; C) \Big\} \Big) + \alpha C(t).
\end{align}
\end{subequations}

Observe that $\grad\tilde{\Delta}_{\rHS}$ is continuous on $\Uc$, and
therefore bounded. It is shown in \cite{Dominy08a} that $\big( \rmd_{C}
U_{\tf} \big)^{*}$ is uniformly bounded for all $C \in \Kbb_{\tf}$. In
fact, $\big\| \big( \rmd_{C} U_{\tf} \big)^{*} \big\|_{\infty} \leq \| \mu
\|_{\rHS} \| \eta \|_{L^{1}}^{\frac{1}{2}}$, where $\| \cdot \|_{L^{1}}$
denotes the $L^{1}$ norm.\footnote{Here, $\big\| \big( \rmd_{C} U_{\tf}
\big)^{*} \big\|_{\infty}$ denotes the operator norm \cite{Horn90a}:
$\displaystyle{\big\| \big( \rmd_{C} U_{\tf} \big)^{*} \big\|_{\infty}
:=  \sup_{\substack{\delta U \in T \Uc \\ \delta U \neq 0}}
\frac{\left\| \big( \rmd_{C} U_{\tf})^{*} (\delta U \big)
\right\|_{\Kbb_{\tf}}}{\| \delta U \|_{\rHS}}}$.} Therefore, there
exists a $\kappa > 0$ such that
\begin{equation}
\left\| \big( \rmd_{C} U_{\tf} \big)^{*} \big( \grad
\Delta_{\rHS}^{[V]} [U_{\tf}(C)] \big) \right\| \leq \kappa
\end{equation}
for all $C \in \Kbb_{\tf}$. Then, since 
\begin{equation}
C(t) = -\frac{\eta(t)}{4\alpha n \,\Delta_{\rHS}^{[V]} [U_{\tf}(C)]}
\Imag \Big( \Tr \Big\{ \big[ R^{\dag} U_{\tf}(C) - U^{\dag}(\tf; C) R
\big] U^{\dag}(t; C) \mu U(t; C) \Big\} \Big)
\end{equation}
at a local extremum of $\J$, all local extrema of $\J$ must lie within
the hypersphere $\| C \| \leq \kappa/\alpha$. From this perspective,
the role of the shape function can be interpreted as moving
undesirably-shaped controls outside of this hypersphere, to the region
where they can no longer act as local extrema of $\J$. Note that both
the space $K_{\tf}$ and the control objective $\J$ depend on the shape
function $\eta(t)$. When a shape function is used in an iterative
optimization routine via (\ref{eqa:gradient}), this results in the
minimization of a different (though conceptually related) objective on a
different space than when $\eta(t) \equiv 1$.

\section*{References}
\bibliography{Grace_Distance_NJP.bib}
\bibliographystyle{iopart-num}

\end{document}